\documentclass{article}
\pdfoutput=1

\usepackage{fullpage}
\usepackage{color}
\usepackage{graphicx}
\usepackage{amsmath}
\usepackage{amssymb}
\usepackage{amsthm}
\usepackage[noend]{algpseudocode}
\usepackage{algorithm}
\usepackage{caption,subcaption}

\title{A Faster Algorithm for Computing Straight Skeletons}
\author{Siu-Wing Cheng \and Liam Mencel \and Antoine Vigneron}

\newcounter{theoremcounter}

\newtheorem{theorem}[theoremcounter]{Theorem}
\newtheorem{lemma}[theoremcounter]{Lemma}

\newcommand{\ske}{\mathcal{S}}
\newcommand{\poly}{\mathcal{P}}
\newcommand{\subdiv}{\mathcal{K}(\poly)}
\newcommand{\cell}{\mathcal{C}}
\newcommand{\terrain}{\mathcal{T}}
\newcommand{\motor}{\mathcal{G}}

\begin{document}

\maketitle

\begin{abstract}
We present a new algorithm for computing the straight skeleton of a polygon. For a polygon with $n$ vertices, among which $r$ are reflex vertices, we give a deterministic algorithm that reduces the straight skeleton computation to a motorcycle graph computation in $O(n (\log n)\log r)$ time. It improves on the previously best known algorithm for this reduction, which is randomized, and runs in expected $O(n \sqrt{h+1}\log^2 n)$ time for a polygon with $h$ holes.  Using  known motorcycle graph algorithms, our result yields improved time bounds for computing straight skeletons. In particular, we can compute the straight skeleton of a non-degenerate polygon in $O(n (\log n) \log r + r^{4/3+\varepsilon})$ time for any $\varepsilon>0$. On degenerate input, our time bound increases to $O(n (\log n) \log r + r^{17/11+\varepsilon})$.
\end{abstract}

\section{Introduction}
The straight skeleton of a polygon is defined as the trace of the vertices when the polygon shrinks, each edge moving at the same speed inwards in a perpendicular direction to its orientation.  (See \figurename~\ref{fig:shrink}.) It differs from the medial axis~\cite{medial} in that it is a straight line graph embedded in the original polygon, while the medial axis may have parabolic edges. The notion was introduced by Aichholzer et al.~\cite{aichholzer} in 1995, who gave the earliest algorithm for computing the straight skeleton. However, the concept has been recognized as early as 1877 by von Peschka~\cite{peschka}, in his interpretation as projection of roof surfaces.

\begin{figure}	
	\begin{subfigure}[b]{.3\textwidth}
	\centering \includegraphics[width=.8\textwidth]{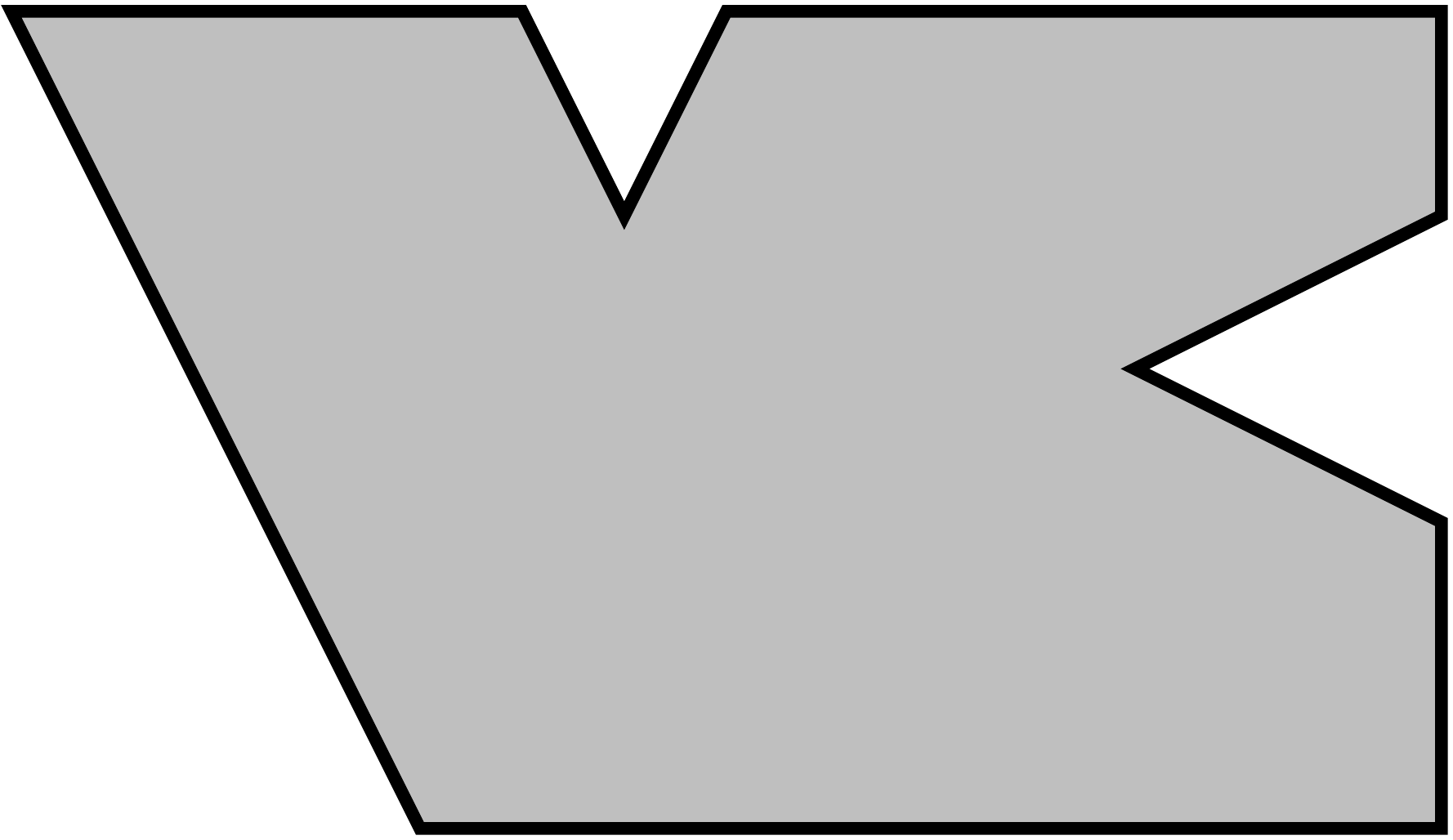}
	\caption{The input polygon $\poly$.\label{fig:shrink1}}
	\end{subfigure}
	\hspace{1ex}
	\begin{subfigure}[b]{.3\textwidth}
	\centering \includegraphics[width=.8\textwidth]{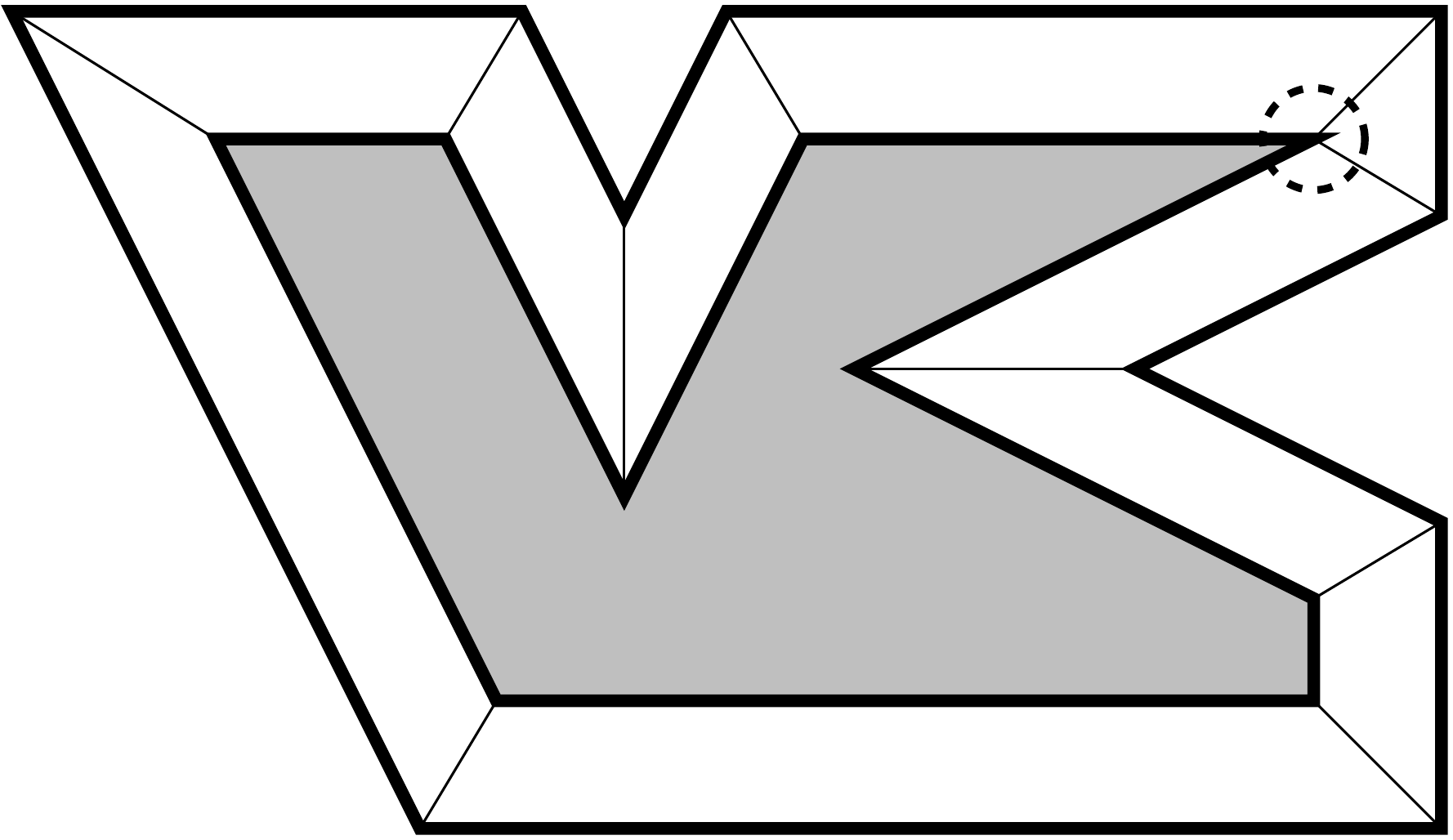}
	\caption{An offset of $\poly$.\label{fig:shrink2}}
	\end{subfigure}
	\hspace{1ex}
	\begin{subfigure}[b]{.3\textwidth}
	\centering \includegraphics[width=.8\textwidth]{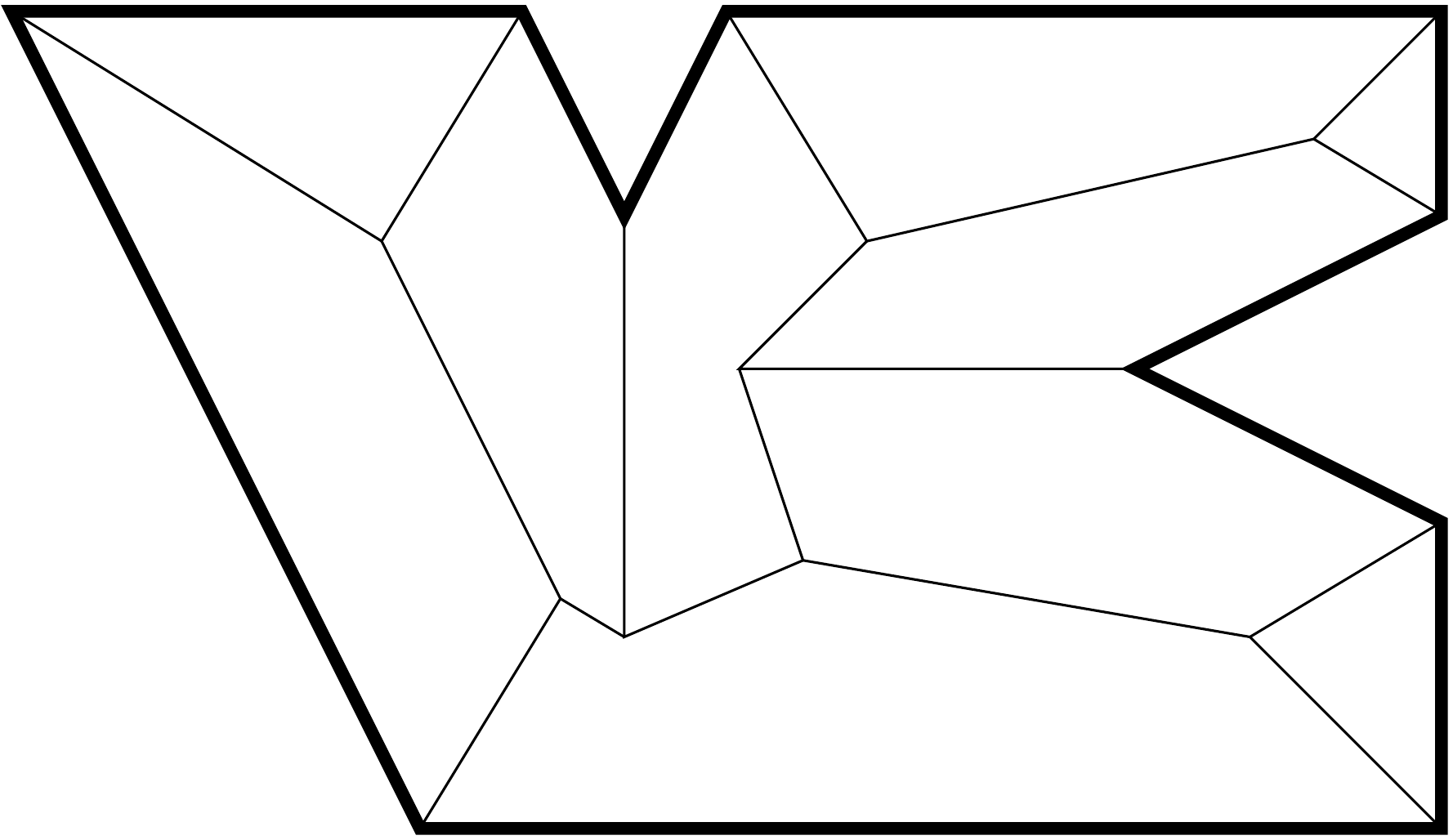}
	\caption{Straight skeleton $\ske$.\label{fig:shrink3}}
	\end{subfigure}
	\caption{The straight skeleton is obtained by shrinking the input polygon $\poly$.\label{fig:shrink}}
\end{figure}

The straight skeleton has numerous applications in computer graphics. It allows one to compute offset polygons~\cite{eppstein}, which is a standard operation in CAD. Other applications include architectural modelling~\cite{arch}, biomedical image processing~\cite{bio}, city model reconstruction~\cite{city}, computational origami~\cite{origami1,origami2,origami3} and polyhedral surface reconstruction~\cite{psr1,psr3,psr2}. Improving the efficiency of straight skeleton algorithms increases the speed of related tools in geometric computing.

The first algorithm by Aichholzer et al.~\cite{aichholzer} runs in $O(n^2 \log n)$ time, and simulates the shrinking process discretely. Eppstein and Erickson~\cite{eppstein} developed the first sub-quadratic algorithm, which runs in $O(n^{17/11 + \varepsilon})$ time. In their work, they proposed motorcycle graphs as a means of encapsulating the main difficulty in computing straight skeletons. Cheng and Vigneron~\cite{cheng} expanded on this notion by reducing the straight skeleton problem in non-degenerate cases to a motorcycle graph computation and a lower envelope computation. This reduction was later extended to degenerate cases by Held and Huber~\cite{huber2}. Cheng and Vigneron gave an algorithm for the lower envelope computation on a non-degenerate polygon with $h$ holes, which runs in $O(n \sqrt{h+1} \log^2 n)$ expected time. They also provided a method for solving the  motorcycle graph problem in $O(n \sqrt{n} \log n)$ time. Putting the two together gives an algorithm which solves the straight skeleton problem in $O(n \sqrt{h+1} \log^2 n + r \sqrt{r} \log r)$ expected time, where $r$ is the number of reflex vertices. 

\paragraph{Comparison with previous work.}
Recently, Vigneron and Yan \cite{yan} found a faster, $O(n^{4/3+\varepsilon})$-time algorithm for computing motorcycle graphs. It thus removed one bottleneck in straight skeleton computation. In this paper we remove the second bottleneck: We give a faster reduction to the motorcycle graph problem. Our algorithm performs this reduction in  deterministic $O(n (\log n)\log r)$ time, improving on the previously best known algorithm, which is randomized and runs in expected $O(n \sqrt{h+1}\log^2 n)$ time~\cite{cheng}. Recently, Bowers independently discovered an $O(n \log n)$-time, deterministic algorithm to perform this reduction in the case of simple polygons, using a very different approach~\cite{Bowers14}.

Using known algorithms for computing motorcycle graphs, our reduction yields faster algorithms for computing the straight skeleton. In particular, using the algorithm by Vigneron and Yan~\cite{yan}, we can compute the straight skeleton of a non-degenerate polygon in $O(n(\log n) \log r + r^{4/3+\varepsilon})$ time for any $\varepsilon>0$. On degenerate input, we use Eppstein and Erickson's algorithm, and our time bound increases to $O(n(\log n) \log r + r^{17/11+\varepsilon})$. For simple polygons whose coordinates are $O(\log n)$-bit rational numbers, we can compute the straight skeleton in $O(n \log^2 n)$ time using the motorcycle graph algorithm by Vigneron and Yan~\cite{yan} (even in degenerate cases). Table~\ref{tab:bounds} summarizes the previously known results and compares with our new algorithm. $O^*$ denotes the expected time bound of a randomized algorithm, and $O$ is for deterministic algorithms. To make the comparison easier, we replaced the number of holes $h$ with $r$, as $h=O(r)$.

\begin{table}
\centering
\begin{tabular}{|l|l|l|}
\hline
		&	Previously best known	& 	This paper \\
\hline
Arbitrary polygon	& $O(n^{8/11+\varepsilon}r^{9/11+\varepsilon})$ ~ \cite{eppstein}
			& $O(n(\log n) \log r + r^{17/11+\varepsilon})$\\
\hline
Non-degenerate polygon 	& $O^*(n\sqrt{r}\log^2 n)$ ~ \cite{cheng} &	 $O(n (\log n) \log r + r^{4/3+\varepsilon})$\\
\hline
Simple pol., arbitrary  &	$O^*(n\log^2 n+r^{17/11+\varepsilon})$ \ \cite{cheng,eppstein} &	
$O(n(\log n)\log r+r^{17/11+\varepsilon})$	\\
\hline
Simple pol., $O(\log n)$ bits &	$O^*(n\log^2 n)$ ~ \cite{cheng,yan}&	$O(n \log^2 n)$	\\
\hline
\end{tabular}
\caption{Summary of previously best known results, compared with those of our new algorithm.
\label{tab:bounds}}
\end{table}

\paragraph{Our approach.}
We use the known reduction to a lower envelope of slabs in 3D~\cite{cheng,huber2}: First a motorcycle graph induced by the input polygon is computed, and then this graph is used to define a set of slabs in 3D. The lower envelope of these slabs is a terrain, whose edges vertically project to the straight skeleton on the horizontal plane. (See Section~\ref{sec:prelim}.)

The difficulty is that these slabs may cross, and in general their lower envelope is a non-convex terrain, so known algorithms for computing lower envelopes of triangles would be too slow for our purpose. Our approach is thus to remove non-convex features: We compute a subdivision of the input polygon into convex cells such that, above each cell of this subdivision, the terrain is convex. Then the portion of the terrain above each cell can be computed efficiently, as it reduces to computing a lower envelope of planes in 3D. The subdivision is computed recursively, using a divide and conquer approach, in two stages. 

During the first stage (Section~\ref{sec:vertical}), we partition using vertical lines, that is, lines parallel to the $y$-axis. At each step, we pick the vertical line $\ell$ through the median motorcycle vertex in the current cell. We first cut the cell using $\ell$, and we compute the restriction of the terrain to the space above $\ell$, which forms a polyline. It can be computed in near-linear time, as it reduces to computing a lower envelope of line segments in the vertical plane through $\ell$. Then we cut the cell using the steepest descent paths from the vertices of this polyline. (See \figurename~\ref{fig:cut1}.) We recurse until the current cell does not contain any vertex of the motorcycle graph. (See \figurename~\ref{fig:example2_1}.)

The first step ensures that the cells of the subdivision are convex. However, valleys (non-convex edges) may still enter the interior of the cells. So our second stage (Section~\ref{sec:valley}) recursively partitions cells using lines that split the set of valleys of the current cell, instead of vertical lines. (See \figurename~\ref{fig:example2_2}.) As the first stage results in a partition where the restriction of the motorcycle graph to any cell is outerplanar, we can perform this subdivision efficiently by divide and conquer.

Each time we partition a cell, we know which slabs contribute to the child cells, as we know the terrain along the vertical plane through the cutting line. In addition, we will argue via careful analysis that our divide and conquer approach avoids slabs being used in too many iterations, and hence the algorithm completes in $O(n(\log n)\log r)$ time.

We state here the main result of this work:
\begin{theorem}\label{th:main}
Given a polygon $\poly$ with $n$ vertices, $r$ of which being reflex vertices, and given the motorcycle graph induced by $\poly$, we can compute the straight skeleton of $\poly$ in $O(n(\log n)\log r)$ time.
\end{theorem}

Our algorithm does not handle weighted straight skeletons~\cite{eppstein} (where edges move at different speeds during the shrink process), because the reduction to a lower envelope of slabs does not hold in this case.

\section{Notations and preliminaries}
\label{sec:prelim}

The input polygon is denoted by $\poly$. A {\em reflex vertex} of a polygon is a vertex at which the internal angle is more than $\pi$. It has $n$ vertices, among which $r$ are reflex vertices. We work in $\mathbb{R}^3$ with $\poly$ lying flat in the $xy$-plane. The  $z$-axis becomes analogous to the time dimension. We  say that a line, or a line segment, is {\em vertical}, if it is parallel to the $y$-axis, and we say that a plane is vertical if it is orthogonal to the $xy$-plane. The boundary of a set $A$ is denote by $\partial A$. The cardinality of a set $A$ is denoted by $|A|$. We denote by $\overline{pq}$ the line segment with endpoints $p,q$.

\paragraph{Terrain.} At any time, the horizontal plane $z=t$ contains a snapshot of $\poly$ after shrinking for $t$ units of time. While the shrinking polygon moves vertically at unit speed, faces are formed as the trace of the edges, and these faces make an angle  $\pi / 4$ with the $xy$-plane. The surface formed by the traces of the edges is the {\em terrain} $\terrain$. (See \figurename~\ref{fig:slabs} a.) 
The traces of the vertices of $\poly$ form  the set of edges of $\mathcal{T}$. An edge $e$ of $\terrain$ is {\em convex} if there is a plane through $e$ that is above the two faces bounding $e$.  The edges of $\mathcal{T}$ corresponding to the traces of the reflex vertices will be referred to as {\em valleys}.  Valleys are the only non-convex edges on $\terrain$.  The other edges, which are convex, are called {\em ridges}. The {\em straight skeleton} $\ske$ is the graph obtained by projecting the edges and vertices of $\mathcal{T}$ orthogonally onto the $xy$-plane. We also call valleys and ridges the edges of $\ske$ that are obtained by projecting valleys and ridges of $\terrain$ onto the $xy$-plane.

\begin{figure}
\centering
\includegraphics[scale=.7]{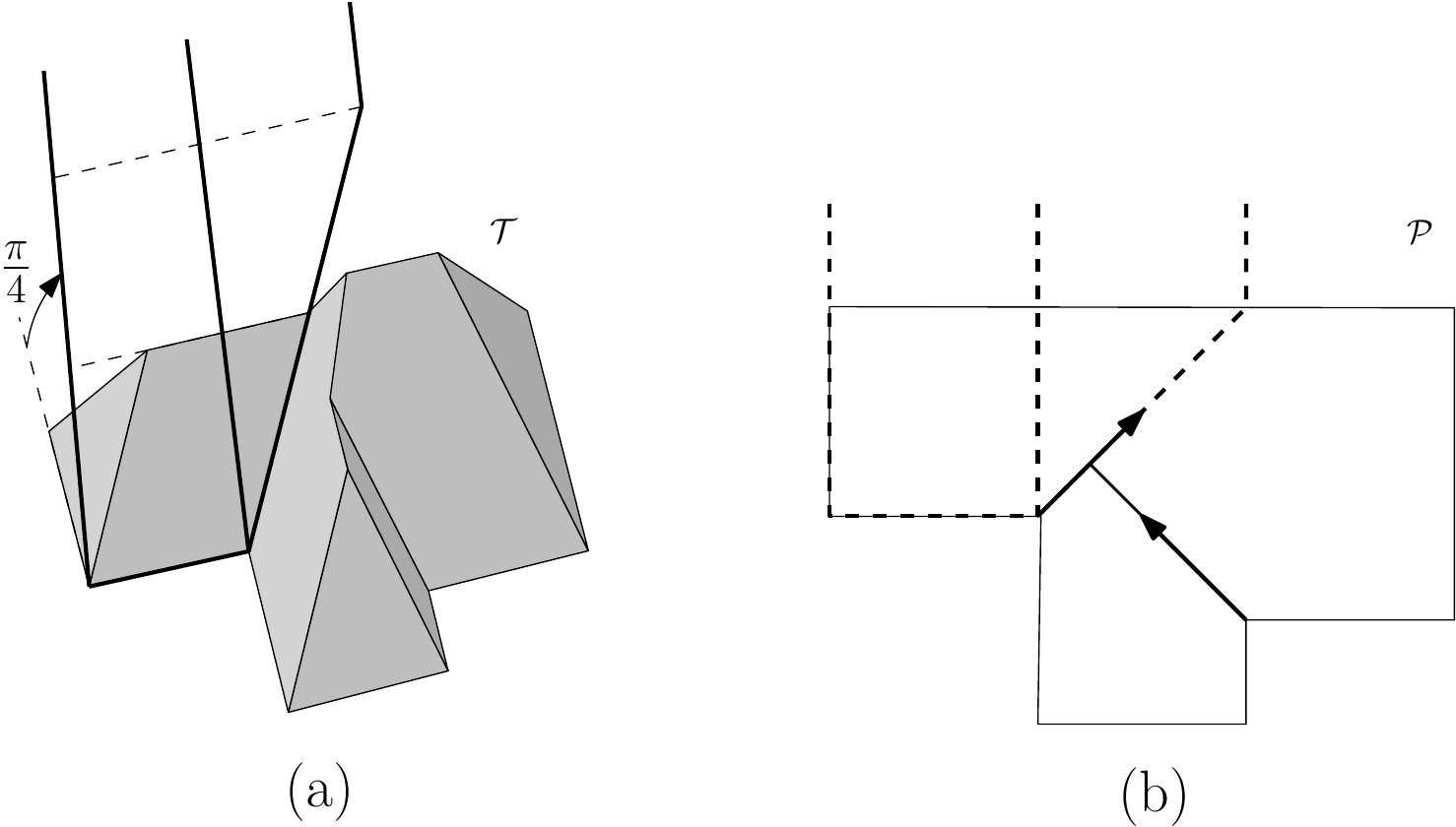}
\caption{Illustration of the two different types of slabs. (a) The terrain $\terrain$, an edge slab and motorcycle slab. This terrain has two valleys, adjacent to the two reflex vertices of the polygon. (b) The motorcycle graph associated with $\mathcal{P}$ and the boundaries of the edge slab and the motorcycle slab viewed from above.\label{fig:slabs}}
\end{figure}

\paragraph{Motorcycle graph.} Our algorithm for computing the straight skeleton assumes that a motorcycle graph induced by $\mathcal{P}$ is precomputed~\cite{cheng}. This graph is defined as follows.
A motorcycle is a point moving at a fixed velocity. We place a motorcycle at each reflex vertex of $\poly$. The velocity of a motorcycle is the same as the velocity of the corresponding reflex vertex when $\poly$ is shrunk, so its direction is the bisector of the interior angle, and its  speed is $1/\sin{(\theta / 2)}$, where $\theta$ is the exterior angle at the reflex vertex. (See \figurename~\ref{fig:motora}.)

The motorcycles begin moving simultaneously. They each leave behind a track as they move. When a motorcycle collides with either another motorcycle's track or the boundary of $\poly$, the colliding motorcycle halts permanently. (In degenerate cases, a motorcycle may also collide head-on with another motorcycle, but for now we rule out this case.) After all motorcycles stop, the tracks form a planar graph called the {\em motorcycle graph induced by $\poly$}. (see \figurename~\ref{fig:motorb}.)

\begin{figure}
	\centering
	\begin{subfigure}[b]{0.45\textwidth}
		\centering\includegraphics[scale=.7]{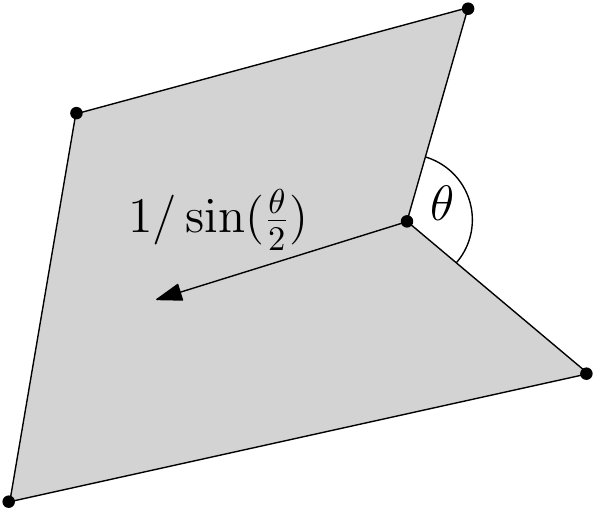}
		\caption{\label{fig:motora}}
	\end{subfigure}	
	\hspace{5ex}
	\begin{subfigure}[b]{0.45\textwidth}
		\centering\includegraphics[scale=.7]{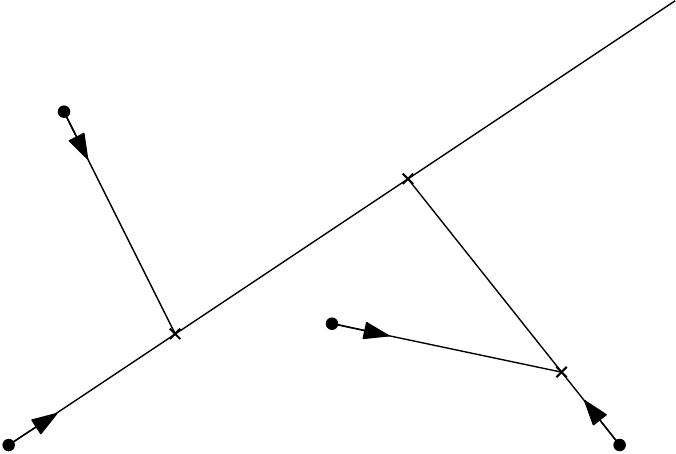}
		\caption{\label{fig:motorb}}
	\end{subfigure}
	\caption{Motorcycle graph.\label{fig:motor}}
\end{figure}

\paragraph{General position assumptions.} To simplify the description and the analysis of our algorithm, we assume that the polygon is in general position. No edge of $\mathcal{P}$ or $\ske$ is vertical. No two motorcycles collide with each other in the motorcycle graph, and thus each valley is adjacent to some reflex vertex. Each vertex in the straight skeleton graph has degree 1 or 3. Our results, however, generalize to degenerate polygons, as explained in Section~\ref{sec:degen}.

\paragraph{Lifting map.} The {\em lifted} version $\hat p$ of a point $p \in \poly$ is the point on $\terrain$ that is vertically above $p$. In other words, $\hat p$ is the point of $\terrain$ that projects orthogonally to $p$ on the $xy$-plane. We may also apply this transformation to a line segment $s$ in the $xy$-plane, then $\hat s$ is a polyline in $\terrain$. 
We will abuse notation and denote by $\hat \motor$ a lifted version of $\motor$ where the height of a point is the time at which the corresponding motorcycle reaches it. Then the lifted version $\hat e$ of an edge $e$ of $\motor$ does not lie entirely on $\terrain$, but it contains the corresponding valley, and the remaining part  of $\hat e$ lies above $\terrain$~\cite{cheng}. (See \figurename~\ref{fig:slabs}a.)

Given a point $\hat p$ that lies in the interior of a face $f$ of $\terrain$, there is a unique steepest descent path from $\hat p$ to the boundary of $\poly$. This path consists either of a straight line segment orthogonal to the base edge $e$ of $f$, or it consists of a segment going straight to a valley, and then follows this valley. (In degenerate cases, the path may follow several valleys consecutively.) If $\hat p$ is on a ridge, then two such descent paths from $p$ exist, and if $\hat p$ is a convex vertex, then there are three such paths. (See \figurename~\ref{fig:canon3}.)

\begin{figure}
	\begin{subfigure}[b]{0.3\textwidth}
	    \centering \includegraphics[width=0.8\textwidth]{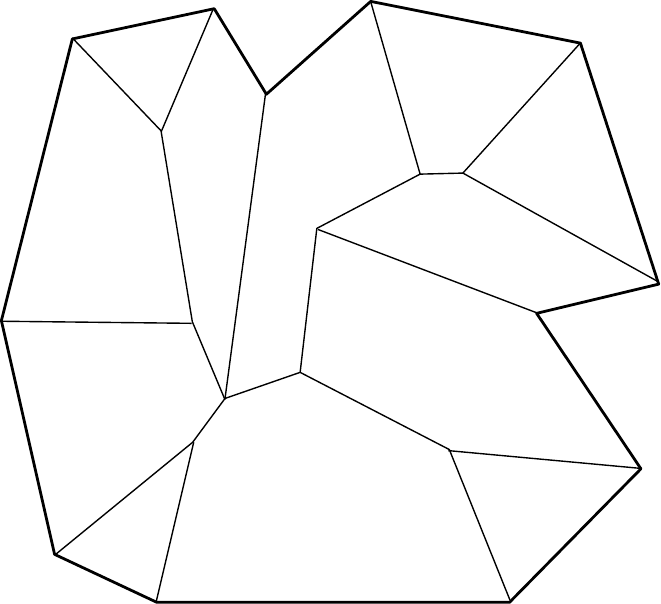}
	    \caption{The skeleton $\ske$.\label{fig:canon1}}
	\end{subfigure}
	\hspace{3ex}
	\begin{subfigure}[b]{0.3\textwidth}
	    \centering \includegraphics[width=0.8\textwidth]{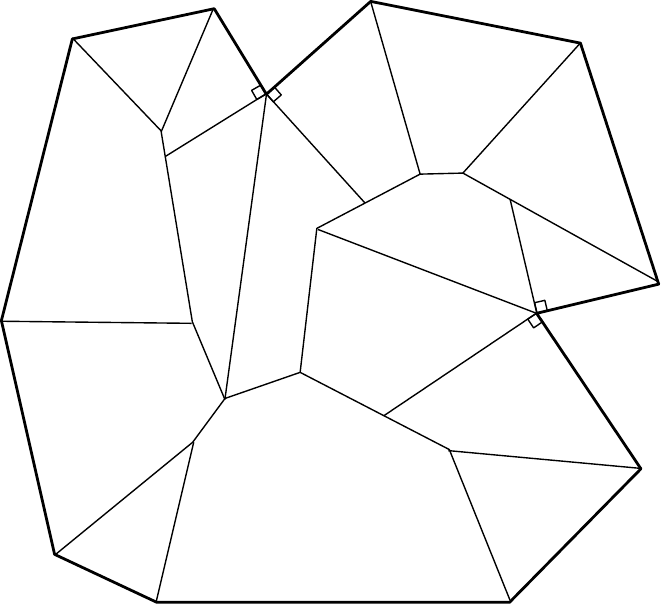}
	    \caption{The skeleton $\ske'$.\label{fig:canon2}}
	\end{subfigure}
	\hspace{3ex}
	\begin{subfigure}[b]{0.3\textwidth}
	    \centering \includegraphics[width=0.8\textwidth]{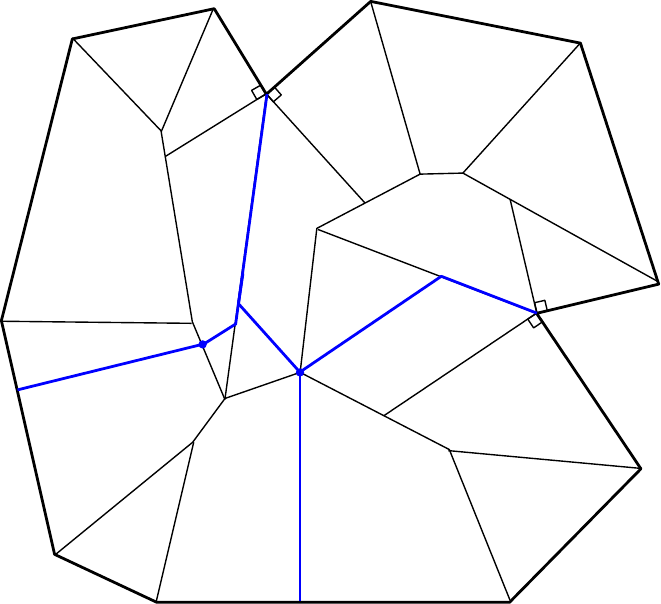}
	    \caption{Descent paths.\label{fig:canon3}}
	\end{subfigure}
	\caption{The polygon $\poly$, its skeletons and descent paths.\label{fig:canon}}
\end{figure}

\paragraph{Reduction to a lower envelope.}
Following Eppstein and Erickson~\cite{eppstein}, Cheng and Vigneron~\cite{cheng}, and Held and Huber~\cite{huber2}, we use a construction of the straight skeleton based on the lower envelope of a set of three-dimensional slabs. Each edge $e$ of $\mathcal{P}$ defines an {\em edge slab}, which is a 2-dimensional half-strip at an angle of $\pi / 4$ to the $xy$-plane, bounded below by $e$ and along the sides by rays perpendicular to $e$. (See \figurename~\ref{fig:slabs}.) We say that $e$ is the {\em source} of this edge slab.
 
For each reflex vertex $v = e \cap e'$, where $e$ and $e'$ are edges of $\poly$, we define two {\em motorcycle slabs} making angles of $\pi / 4$ to the $xy$-plane. One motorcycle slab is bounded below by the edge of $\hat{\motor}$ incident to $v$ and is bounded on the sides by two rays from each end of this edge in the ascent direction of $e$. The other motorcycle slab is defined similarly with $e$ replaced by $e'$. The {\em source} of a motorcycle slab is the corresponding edge of $\hat \motor$.
Cheng and Vigneron~\cite{cheng} proved the following result, which was extended to degenerate cases by Huber and Held~\cite{huber}:
\begin{theorem}
\label{thm:character}
The terrain $\mathcal{T}$ is the restriction of the lower envelope of the edge slabs and the motorcycle slabs to the space vertically above the polygon.
\end{theorem}

Our algorithm constructs a graph $\ske'$, which is obtained from $\ske$ by adding two edges at each reflex vertex $v$ of $\poly$ going inwards and orthogonally to each edge of $\poly$ incident to $v$. (See \figurename~\ref{fig:canon2}.) These extra edges are called {\em flat edges}.
We also include the edges of $\poly$ into $\ske'$. It means that each face $f$ of $\ske'$ corresponds to exactly one slab. More precisely, a face is the vertical projection of $\terrain \cap \sigma$ to the $xy$-plane for some slab $\sigma$. By contrast, in the original straight skeleton $\ske$, a face incident to a reflex vertex corresponds to one edge slab and one motorcycle slab.

\section{Computing the vertical subdivision}
\label{sec:vertical}

In this section, we describe and we analyze the first stage of our algorithm, where the input polygon $\poly$ is recursively partitioned using vertical cuts. The corresponding procedure is called {\scshape Divide-Vertical}, and its pseudocode can be found in Algorithm~\ref{alg:vertical}. It results in a subdivision of the input polygon $\poly$, such that any cell of this subdivision has the following property: It does not contain any vertex of $\motor$ in its interior, or it is contained in the union of two faces of $\ske'$. The second stage of our algorithm is presented in Section~\ref{sec:valley}.

\subsection{Subdivision induced by a vertical cut}\label{sec:verticalcut}

At any step of the algorithm, we maintain a planar subdivision $\subdiv$, which is a partition of the input polygon $\poly$ into polygonal cells. Each cell is a polygon, hence it is connected. A cell $\cell$ in the current subdivision $\subdiv$ may be further subdivided as follows.

Let $\ell$ be a vertical line through a vertex of $\motor$. We assume that $\ell$ intersects $\cell$, and hence $\cell \cap \ell$ consists of several line segments $s_1,\dots,s_q$. These line segments are introduced as new boundary edges in $\subdiv$; they are called the {\em vertical edges} of $\subdiv$. They may be further subdivided during the course of the algorithm, and we still call the resulting edges vertical edges. 

We then insert non-vertical edges along steepest descent paths, as follows. Note that we are able to efficiently compute the intersection $\ske' \cap \ell$ without knowing $\ske'$, this is explained in the detailed description of the algorithm. Each intersection point $p \in s_j \cap \ske'$ has a lifted version $\hat p$ on $\terrain$. By our non-degeneracy assumptions, there are at most three steepest descent paths to $\partial \cell$ from $\hat p$. The vertical projections of these paths onto $\cell$ are also inserted as new edges in $\subdiv$. The resulting partition of $\cell$ is the {\em subdivision induced by} $\ell$. (See \figurename~\ref{fig:example}.)

We denote by $\cell_1,\cell_2,\dots$ the cells of $\subdiv$ that are constructed during the course of the algorithm. Let $\ell_i^-$ and $\ell_i^+$ denote the vertical lines through the leftmost and rightmost point of $\cell_i$, respectively. When we perform one step of the subdivision, each new cell lies entirely to the left or to the right of the splitting line, and thus by induction, any vertical edge of a cell $\cell_i$ either lies in $\ell_i^-$ or $\ell_i^+$. We now study the geometry of these cells.
\begin{lemma}\label{lem:convexity} Let $p$ be a reflex vertex of a cell $\cell_i$. Then $p$ is a reflex vertex of $\poly$ such that $\partial \cell_i$ and $\partial \poly$ coincide in a neighborhood of $p$, or $p$ is a point where a descent path bounding $\cell_i$ reaches a valley.
\end{lemma}
\begin{proof}
We prove it by induction. The initial cell is $\cell_1=\poly$, and hence the property holds. When we perform a subdivision of a cell $\cell_i$ along a line $\ell$, we cannot introduce reflex vertices along $\ell$, as we insert the segments $\cell_i \cap \ell$ as new cell boundaries. So new reflex vertices may only appear along descent paths. They cannot appear at the lower endpoint of a descent path, as a descent path can only meet a reflex vertex along its exterior angle bisector.  So a reflex vertex may only appear in the interior of a descent path, and a descent path only bends when it reaches a valley.
\end{proof}

The lemma above shows that non-convexity may only be introduced when a bounding path reaches a valley. The lemma below implies that, at any point in time, it can occur only once per valley. (See \figurename~\ref{fig:example}.)
\begin{lemma}\label{lem:valley}
Let $e=\overline{pq}$ be a valley or a flat edge of $\ske'$, with $p$ being a reflex vertex of $\poly$ and $q$ being the other endpoint of $e$. At any time during the course of the algorithm, there is a point $a$ along $e$ such that $\overline{pa}$ is contained in the union of the boundaries of the cells of $\subdiv$, and the interior of $\overline{aq}$ is contained in the interior of a cell $\cell_i$.
\end{lemma}
\begin{proof} 
We proceed by induction, so we assume that at the current point of the execution of the algorithm,  there is a point $a$ on $e$ such that $\overline{pa}$ is contained in the union of the edges of $\subdiv$, and $\overline{aq}$ is contained in the interior of a cell $\cell_j$. So $e$ can only intersect  the interior of a new cell if this cell is obtained by subdividing $\cell_j$. When performing this subdivision, at most two descent paths and one vertical cut can intersect $\overline{aq}$, and then the descent paths from these intersection points to $a$ are added as cell boundaries. After that, we are again in the situation where $e$ is split into two segments $\overline{pb}$ and $\overline{bq}$, with $\overline{pb}$ being covered by edges of $\subdiv$ and $\overline{bq}$ being in the interior of a cell.
\end{proof}

A ridge, on the other hand, can cross the interior of several cells. But its intersection with any given cell is a single line segment:
\begin{lemma}\label{lem:ridge}
For any ridge $e$ and any cell $\cell_i$, the intersection $e \cap \cell_i$ is a single line segment, and $e \cap \partial \cell_i$ consists of at most two points.
\end{lemma}
\begin{proof}
As $e$ is a convex edge, the only descent paths that can meet $e$ are descent paths that start from $e$. So $e$ can only be partitioned by a vertical line cut through its interior.  When we perform one such subdivision along a segment of $e$, it is split into two segments, one on each side of the cutting line, and these segments now belong to two different cells. When we repeat the process, it remains true that $e \cap \cell_i$ is a segment, and that it can only meet $\partial \cell_i$ at its endpoints.
\end{proof}

An {\em empty cell} is a cell of $\subdiv$ whose interior does not overlap with $\ske'$. (See \figurename~\ref{fig:emptycells}.) Thus an empty cell is entirely contained in a face of $\ske'$. 
\begin{figure}
	\begin{subfigure}[b]{.45\textwidth}
	\centering
	\centering \includegraphics[width=\textwidth]{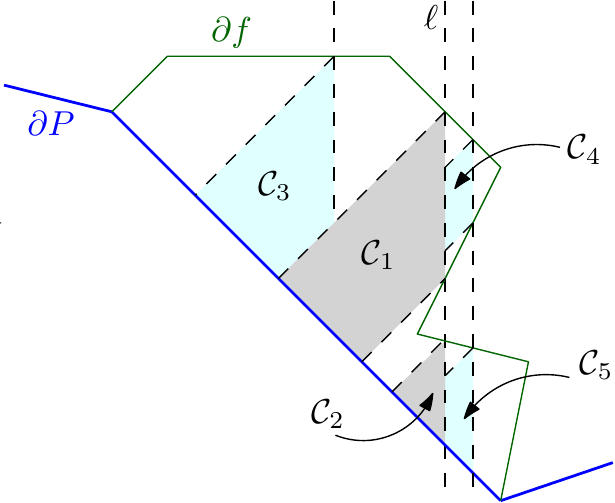}
	\caption{The cells $\cell_1,\dots,\cell_5$ are empty. The first cut is performed along $\ell$.
	\label{fig:emptycells}}
	\end{subfigure}
	\hspace{1ex}
	\begin{subfigure}[b]{.52\textwidth}
	\centering \includegraphics[width=\textwidth]{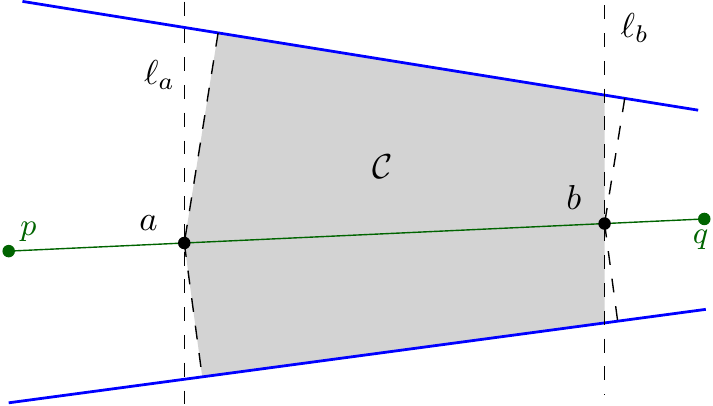}
	\caption{The wedge $\cell$ corresponding to $\overline{ab}$.
	\label{fig:wedge}}
	\end{subfigure}
	\caption{Empty cells  and a wedge.}
\end{figure} 
Another type of cell, called a {\em wedge}, will play an important role in the analysis of our algorithm.  Let $\overline{pq}$ be a ridge of $\ske'$, and let $a,b$ be two points in the interior of $\overline{pq}$. Let $\ell_a$ and $\ell_b$ be the vertical lines through $a$ and $b$, respectively. Consider the subdivision of $\poly$ obtained by inserting vertical boundaries along $\ell_a$ and $\ell_b$, and the four descent paths from $a$ and $b$. (See \figurename~\ref{fig:wedge}.) The cell of this subdivision containing $\overline{ab}$ is called the wedge corresponding to $\overline{ab}$. The lemma below shows that wedges are the only cells that can overlap the interior of a ridge, without enclosing any of its endpoints.

\begin{lemma}\label{lem:wedge}
Let $\cell_i$ be a cell overlapping a ridge, but not its endpoints. Then $\cell_i$ is a wedge.
\end{lemma}
\begin{proof}
Let $a$ and $b$ be the points on $\partial \cell_i$ which are farthest along the ridge in either direction. A ridge can only meet descent paths that start from it, so $a$ and $b$ must each lie on a vertical cut, $\ell_a$ and $\ell_b$. No vertical cut has been made between $a$ and $b$, otherwise $a$ and $b$ could not be in the same cell. So there is no vertical cut in the interior of the wedge corresponding to $\overline{ab}$, and thus no descent path has been traced inside this wedge. It follows that this wedge is $\cell_i$.
\end{proof}

\subsection{Data structure}\label{sec:datastructure}
During the course of the algorithm, we maintain the polygon $\poly$ and its subdivision $\subdiv$ in a doubly-connected edge list~\cite{4M}. So each cell $\cell_i$ is represented by a circular list of edges, or several if it has holes. In the following, we show how we augment these chains so that they record incidences between the boundary of $\cell_i$ and the faces of $\ske'$.  

For each cell $\cell_i$, let $\ske'_i$ be the subdivision of $\cell_i$ induced by $\ske'$. So the faces of $\ske'_i$ are the connected components of $\cell_i \setminus \ske'$. Let $Q$ denote a circular list of edges that form one component of $\partial \cell_i$. We subdivide each vertical edge of $Q$ at each intersection point with an edge of $\ske'$. Now each edge $e$ of $Q$ bounds exactly one face $f_j$ of $\ske'_i$. We store a pointer from $e$ to the slab $\sigma_j$ corresponding to $f_j$. In addition, for each vertex of $Q$ which is a reflex vertex of $\poly$, we store pointers to the two corresponding motorcycle slabs. We call this data structure a {\em face list}. So we store one face list for each connected component of $\partial \cell_i$.
(See \figurename~\ref{fig:facelists}.)

\begin{figure}
	\centering \includegraphics[width=\textwidth]{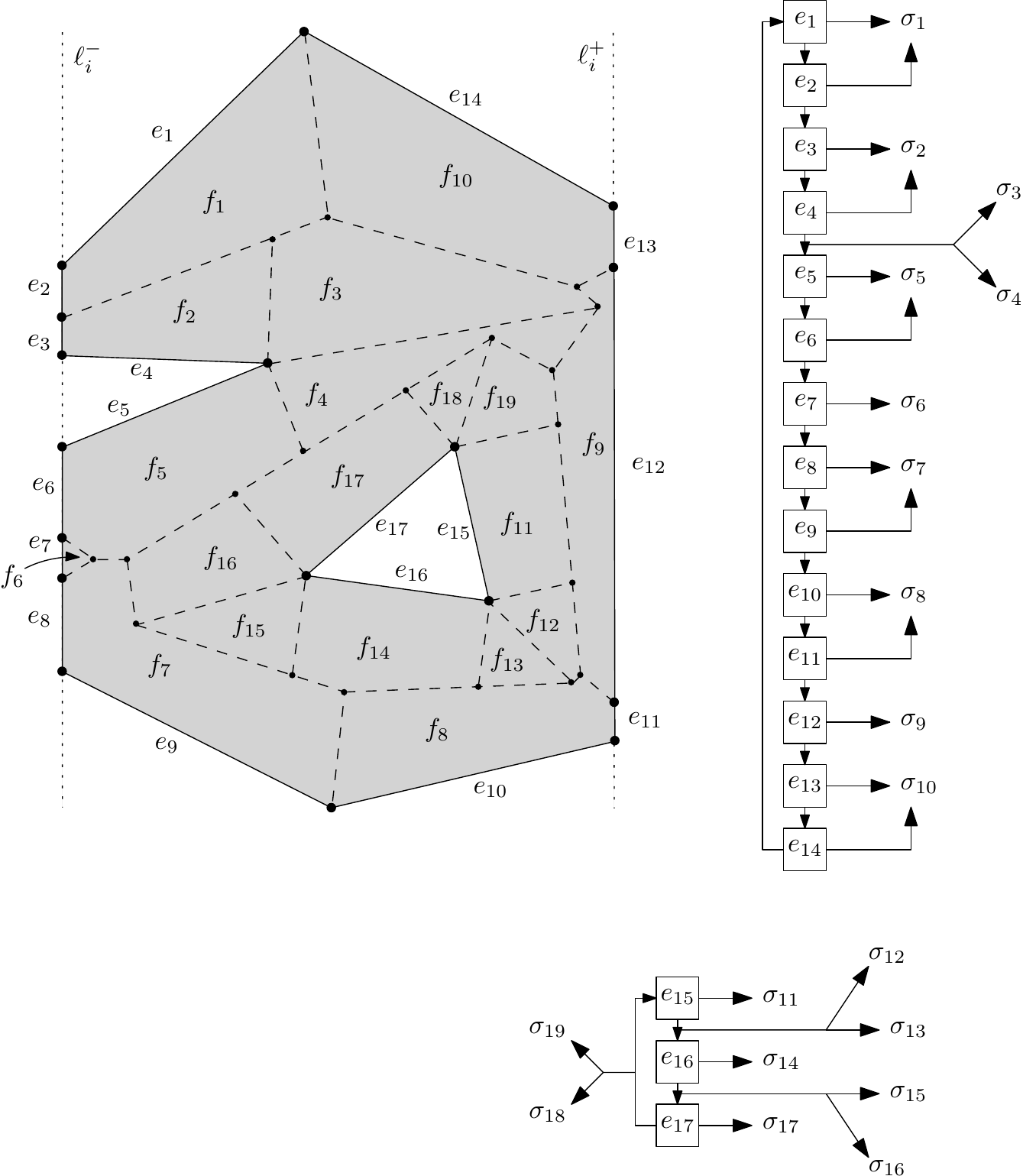}
	\caption{The face lists for the cell $\cell_i$ bounded by the vertical line cuts $\ell_i^-$ and
	$\ell_i^+$. The faces are denoted by $f_1,\dots,f_{19}$ and the corresponding slabs are 
	$\sigma_1,\dots,\sigma_{19}$. The face lists point to these slabs, as the exact shape of the
	faces of $\ske'$ is not known.\label{fig:facelists}}
\end{figure}

Lemma~\ref{thm:holes} makes an observation that will be used in subsequent lemmas.

\begin{lemma}
\label{thm:holes}
A hole of $\cell_i$ is necessarily a hole of $\poly$.
\end{lemma}

\begin{proof}
When we subdivide $\cell_i$, each newly added edge either connects directly to $\partial \cell_i$, or connects to $\partial \cell_i$ via a descent path. Since a hole is a connected component of the boundary of the cell, it follows that no new holes are create in the algorithm. The initial cell $\cell_1$ contains holes which are precisely the holes of $\poly$. 
\end{proof}

We say that a vertex $v$ of the motorcycle graph $\motor$ {\em conflicts} with a cell $\cell_i$ of $\subdiv$ if either $v$ lies in the interior of $\cell_i$, or $v$ is a reflex vertex  of $\partial \cell_i$. We also store the list of all the vertices conflicting with each cell $\cell_i$. This list $V_i$ is called the {\em vertex conflict list} of $\cell_i$. The size of this list is denoted by $v_i$.
In summary, our data structure consists of:
\begin{itemize}
\item A doubly-connected edge list storing $\subdiv$.
\item The face lists and the vertex conflict list $V_i$ of each cell $\cell_i$.
\end{itemize}

We say that an edge $e$ of $\ske'$ conflicts with the cell $\cell_i$ if it intersects the interior of $\cell_i$. So any edge of $\ske'_i$ that is not on $\partial \cell_i$ is of the form $e \cap \cell_i$ for some edge $e$ of $\ske'$ conflicting with $\cell_i$. We denote by $c_i$ the number of edges conflicting with $\cell_i$. During the course of the algorithm, we do not necessarily know all the edges conflicting with a cell $\cell_i$, and we don't even know $c_i$, but this quantity will be useful for analyzing the running time. In particular, it allows us to bound the size of the data structure for $\cell_i$.

\begin{lemma}\label{lem:structuresize}
If $\cell_i$ is non-empty, then the total size of the face lists of $\cell_i$ is $O(c_i)$. In particular, it implies that $\partial \cell_i$ has $O(c_i)$ edges, and $\cell_i$ overlaps $O(c_i)$ faces of $\ske'$. On the other hand, if $\cell_i$ is empty, then the total size is $O(1)$, and thus $\partial \cell_i$ has $O(1)$ edges.
\end{lemma}
\begin{proof} Let $Q$ denote the outer boundary of $\cell_i$, and let $|Q|$ denote its number of edges. By Lemma~\ref{lem:convexity}, each reflex vertex $p$ of $Q$ is in a valley, and the two edges of $Q$ incident to $p$ bound the two faces of $\ske'_i$ incident to this valley. So any subchain $Q'$ of $Q$ that bounds only one face $f'$ of $\ske'_i$ must be convex. The edges of $Q'$ can take only 3 directions: vertical, parallel to the base edge of $f$, or the steepest descent direction. So $Q'$ can have at most 5 edges: two vertical edges, two edges parallel to the steepest descent direction, and one edge along the base edge of $f'$.

Thus, $Q$ can be partitioned into at least $|Q|/5$ subchains, such that two consecutive subchains bound different faces. Any vertex of $Q$ at which two consecutive subchains meet must be incident to an edge $e$ of $\ske'_i$ that conflicts with $\cell_i$. By Lemma~\ref{lem:valley} and~\ref{lem:ridge}, this edge can meet $\partial \cell_i$ at most twice. So in total, $Q$ has at most $10(c_i+1)$ edges.

Now consider the holes of $\cell_i$, if any. Such a hole must be a hole of $\poly$ according to Lemma~\ref{thm:holes}, so each vertex along its boundary is the endpoint of at least one edge that conflicts with $\cell_i$. Each conflicting edge is adjacent to at most one hole vertex, so there are $O(c_i)$ such vertices in $\cell_i$. In addition, each edge of  a hole bounds only one face, and for each reflex vertex, another two faces corresponding to motorcycle slabs are added. So in total, the face lists for holes have size $O(c_i)$.

We just proved that the total size of the face lists is $O(c_i+1)$. If $c_i$ is non-empty, we have $c_i\geq 1$, and thus the bound can be written $O(c_i)$. Otherwise, if $\cell_i$ is empty, then it does not conflict with any edge, so $c_i=0$. Hence, the data structure has size $O(1)$.
\end{proof}

\subsection{Algorithm}\label{sec:AlgoVertical}
Our algorithm partitions $\poly$ recursively, using vertical cuts, as in Sect.~\ref{sec:verticalcut}. In this section, we show how to perform a step of this subdivision in  near-linear time. A cell $\cell_i$ is subdivided along a vertical cut through its median conflicting vertex, so the vertex conflict lists of the new cells will be at most half the size of the conflict lists of $\cell_i$. When the vertex conflict list of $\cell_i$ is empty, we call the procedure {\scshape Divide-Valley} presented in Section~\ref{sec:valley}. If $\cell_i$ is empty or is a wedge, then we stop subdividing $\cell_i$, and it becomes a {\em leaf cell}. 

We now describe in more details how we perform this subdivision efficiently. We assume that the cell $\cell_i$ conflicts with at least one vertex, and  that $\cell_i$ is given with the corresponding data structure as described in Sect.~\ref{sec:datastructure}.  We first find the median conflicting vertex in time $O(v_i)$. We compute the list of vertical boundary segments $s_1,\dots,s_q$ created by the cut along the vertical line $\ell$ through the median vertex. This list is sorted along $\ell$, and it can be constructed in time proportional to the number of edges bounding $\cell_i$, which is $O(c_i)$ by Lemma~\ref{lem:structuresize}.

Then we compute the lifted polylines $\hat s_1,\dots,\hat s_q$ as follows. Let $H$ denote the vertical plane through $\ell$.  We first find the list of slabs corresponding to the faces of $\ske'_i$.  We obtain this list as the union of the slabs that appear in the face lists of $\cell_i$. We  compute the intersection of each such slab with $H$. This gives us a set of $O(c_i)$ segments in $H$, of which we compute the lower envelope. It can be done in $O(c_i \log c_i)$ time using an algorithm by Hershberger~\cite{envelope}. Then we obtain $\hat s_1,\dots,\hat s_q$ by scanning through this lower envelope and the list $s_1,\dots,s_q$. Overall it takes time $O(c_i \log c_i)$ to compute this lower envelope, and it has $O(c_i)$ edges, as each edge of $\ske'_i$ or $\cell_i$ creates at most one vertex along this chain.

The partition induced by $\ell$ is obtained by tracing steepest descent paths from $s_1, \dots, s_q$. For a vertical edge $s_j$, any point where $\hat s_j$ changes direction, when projected onto the horizontal plane, corresponds precisely to a point where $s_j$ intersects   an edge $e$ of $\ske'_i$. At each of these points, we do the following without actually knowing $\ske'_i$. There are at most three steepest descent paths from $a=\hat e \cap \hat s_j$, one for each slab through $a$. Each such descent path consists of one line segment along the slab, followed possibly by another line segment along a valley in the case where the slab is a motorcycle slab. Let $\gamma$ denote one of these descent paths. As we know the slab and the starting point of $\gamma$, we can construct $\gamma$ in constant time. This path $\gamma$ goes all the way to $\partial \poly$, so if necessary, we clip it at $\ell_i^-$ or $\ell_i^+$ to obtain its restriction to $\cell_i$.

These descent paths cannot cross, and by construction they do not cross the vertical boundary edges. Each edge of $\ske'_i$ may create at most three such descent paths, so we create $O(c_i)$ such new descent paths. There are also $O(c_i)$ new vertical edges, so we can update the doubly-connected edge list in time $O(c_i \log c_i)$ by plane sweep. Using an additional $O(v_i \log c_i)$  time, we can update the vertex conflict lists during this plane sweep. The face lists can be updated in overall $O(c_i)$ time by splitting the face lists of $\cell_i$ along the lower endpoints of the new descent paths, and inserting new subchains  along each vertical edge $s_j$, which we obtain directly from $\hat s_j$ in linear time.  So we just proved the following:
\begin{lemma}\label{lem:onesubdivision}
We can compute the subdivision of a non-empty cell $\cell_i$ induced by a line through  its median conflicting vertex, and update our data structure accordingly, in $O((c_i+v_i) \log c_i)$ time.
\end{lemma}

\begin{algorithm}
\caption{Vertical subdivision} \label{alg:vertical}
\begin{algorithmic}[1]
  	\Procedure{Divide-Vertical}{$\cell_i$}
    	\State Select median vertex in $V_i$, and draw the vertical line $\ell$ through it.
  	\State Construct the vertical edges $s_1, \dots ,s_q$ of $\ell \cap \cell_i$.
	\State Compute the lower envelope of the slabs along the vertical plane  through $\ell$.
	\State Construct the lifted version $\hat s_1,\dots,\hat s_q$ of the vertical boundary segments.
	\State Trace within $\cell_i$ the steepest descent paths from each vertex of $\hat s_1, \dots, \hat s_q$.
	\State Update $\subdiv$ using $s_1,\dots,s_q$ and the descent paths as new boundaries.
	\For{each child cell $\cell_j$ of $\cell_i$}
		\State Construct the data structure for $\cell_j$.
	    	\If{$\cell_j$ is a wedge or is empty}
			\State Compute $\ske'_j$ by brute force.

		    \Else
			    \If{$V_j=\emptyset$}
	    			\State Call {\scshape Divide-Valley}$(\cell_i)$
    			\Else
	   				\State Call {\scshape Divide-Vertical}$(\cell_j)$.
				\EndIf
    		\EndIf
	\EndFor
	\EndProcedure
\end{algorithmic}
\end{algorithm}

\subsection{Analysis}
In the previous section, we saw that the vertical subdivision of each cell $\cell_i$ can be obtained in time near-linear in the size of the data structure for $\cell_i$. We now bound the overall running time of the algorithm, so we need to bound the sum $\sum_i c_i+v_i$ over all cells created by {\scshape Divide-Vertical}. 

We use the {\em recursion tree} associated with Algorithm~\ref{alg:vertical}. Each node $\nu$ of this tree represents a cell $\cell_i$, and the child cells of $\cell_i$ are stored at the descendants of $\nu$ in the recursion tree. In particular, the cells stored at the descendants of $\nu$ form a partition of the cell stored at $\nu$. Each time we subdivide a cell $\cell_i$, the conflict list of each new cell has at most half the size of the conflict list of $\cell_i$. As there are at most $2r$ vertices in $\motor$, it follows that:
\begin{lemma}\label{lem:recursiontree}
The recursion tree of {\scshape Divide-Vertical} has depth $O(\log r)$. 
\end{lemma}

The degree of any vertex in $\subdiv$ is at most 5, because there can be at most three descent paths through any point, as well as two vertical edges. It implies that any point of $\poly$ is contained in at most 5 cells at each level of the recursion tree. It follows that:
\begin{lemma}\label{lem:pointincell}
Any point in $\poly$ is contained in $O(\log r)$ cells of $\subdiv$ throughout the algorithm.
\end{lemma}
In particular, if we apply this result to each of the $2r$ vertices of $\motor$, we obtain:
\begin{lemma}
\label{lem:vertexconflict}
Throughout the algorithm, the sum $\sum_i v_i$ of the sizes of the vertex conflict lists  is $O(r \log r)$. 
\end{lemma}

We now bound the total number of conflicts between edges of $\ske'$ and cells of $\subdiv$. 
\begin{lemma}\label{lem:edgeconflict} Throughout the algorithm, each edge $e$ of $\ske'$ conflicts with $O(\log r)$ cells. It follows that $\sum_i c_i=O(n \log r)$.
\end{lemma}
\begin{proof} Let $p,q$ denote the endpoints of $e$. First we assume that $e$ is a ridge. By Lemma~\ref{lem:pointincell}, there are at most $O(\log r)$ cells containing $p$ or $q$, so it remains to bound the number of cells that overlap $e$ but not $\{p,q\}$. By Lemma~\ref{lem:wedge}, these must be wedges. There can only be a wedge along $e$ if at least two vertical cuts through $e$ have been made. When the second such cut is made, the wedge associated with a segment $\overline{ab} \subset e$ is created. Assume without loss of generality that $a$ is between $p$ and $b$. Any wedge is a leaf cell, so in order to create a new wedge along $e$, one must cut with a vertical line through $\overline{pa}$ or $\overline{bq}$. (See \figurename~\ref{fig:wedges}.)
\begin{figure}
	\centering \includegraphics[width=\textwidth]{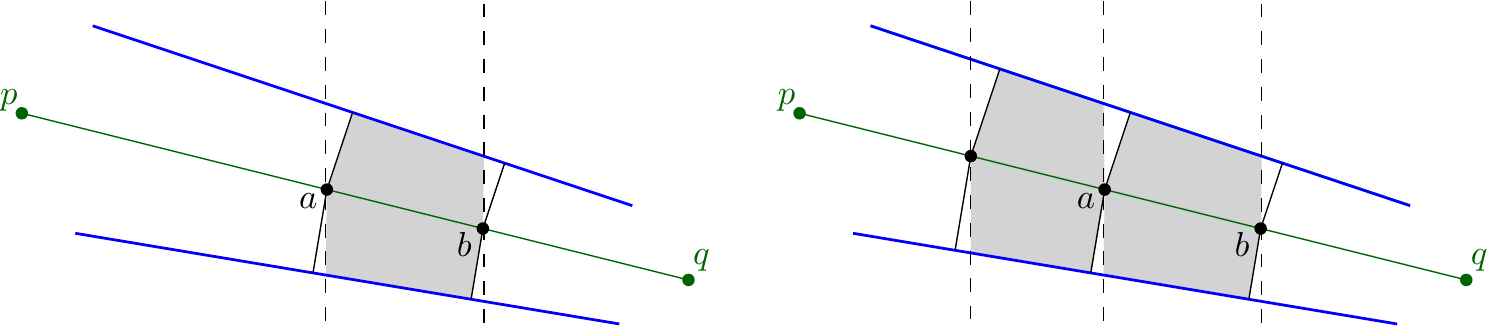}
	\caption{A first wedge is created (left), and an adjacent wedges is created afterwards (right). The cell 	containing $p$ has been split simultaneously.\label{fig:wedges}}
\end{figure}
It creates a new wedge adjacent to the first one, and it splits the cell containing $p$ or $q$, creating a new cell containing $p$ or $q$. Repeating this process, we can see that for each new wedge created along $e$, a new cell containing $p$ or $q$ is created. So there can be only $O(\log r)$ wedges along $e$.

If $e$ is a valley or a flat edge, then by Lemma~\ref{lem:valley}, it only conflicts with cells that contain its higher endpoint, so throughout the algorithm, there are $O(\log r)$ such cells by Lemma~\ref{lem:pointincell}. \end{proof}

We can now state the main result of this section. Its proof follows from Lemma~\ref{lem:structuresize}, \ref{lem:onesubdivision}, Lemma~\ref{lem:vertexconflict}, and \ref{lem:edgeconflict}.
\begin{lemma}\label{lem:vertical}
The vertical subdivision procedure completes in $O(n (\log n) \log r)$ time. The cells of the resulting subdivision are either empty cells, wedges, or do not contain any motorcycle vertex in their interior. They are simply connected, and the only reflex vertices on their boundaries are along valleys.
\end{lemma}
\begin{proof} 
When we perform a subdivision, we can identify in constant time each empty child cell, because by Lemma~\ref{lem:structuresize}, these cells have constant size. When we find such a cell, we do not recurse on it, so these cells do not affect the running time of our algorithm. Therefore, by Lemma~\ref{lem:onesubdivision}, the running  time of Algorithm~\ref{alg:vertical} is the $O(\sum_i (c_i+v_i) \log c_i)$ over all cells created during the course of the algorithm. By Lemma~\ref{lem:vertexconflict} and ~\ref{lem:edgeconflict}, this quantity is $O(n (\log n)\log r)$. The only cells that are not subdivided are empty cells or wedges, hence the other cells cannot contain any motorcycle vertex in their interior. Lemma~\ref{lem:valley} implies the only reflex vertices on the boundary of a cell  are along valleys. 

We prove by contradiction that the cells are simply connected. Suppose that at the end of the vertical subdivision, a cell $\cell_i$ has a hole. This hole must be a hole of $\poly$ according to Lemma~\ref{thm:holes}, hence it has a reflex vertex which conflicts with $\cell_i$. As the conflict list of $\cell_i$ is non-empty, it must be an empty cell or a wedge, in which case it cannot contain a hole of $\poly$.
\end{proof}

\begin{figure}
	\centering
	\begin{subfigure}[b]{0.45\textwidth}
		\centering\includegraphics[width=.7\textwidth]{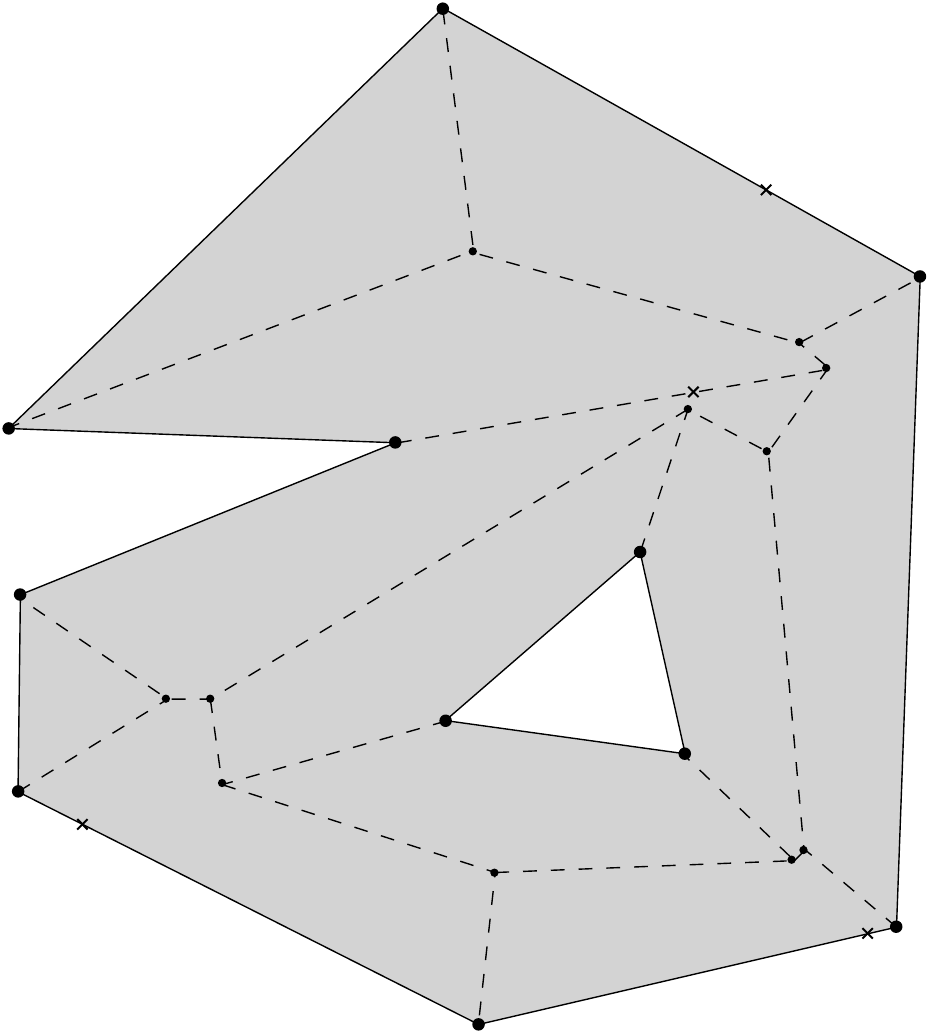}
		\caption{Input polygon and straight skeleton.}
	\end{subfigure}	
	\hspace{5ex}
	\begin{subfigure}[b]{0.45\textwidth}
		\centering\includegraphics[width=.7\textwidth]{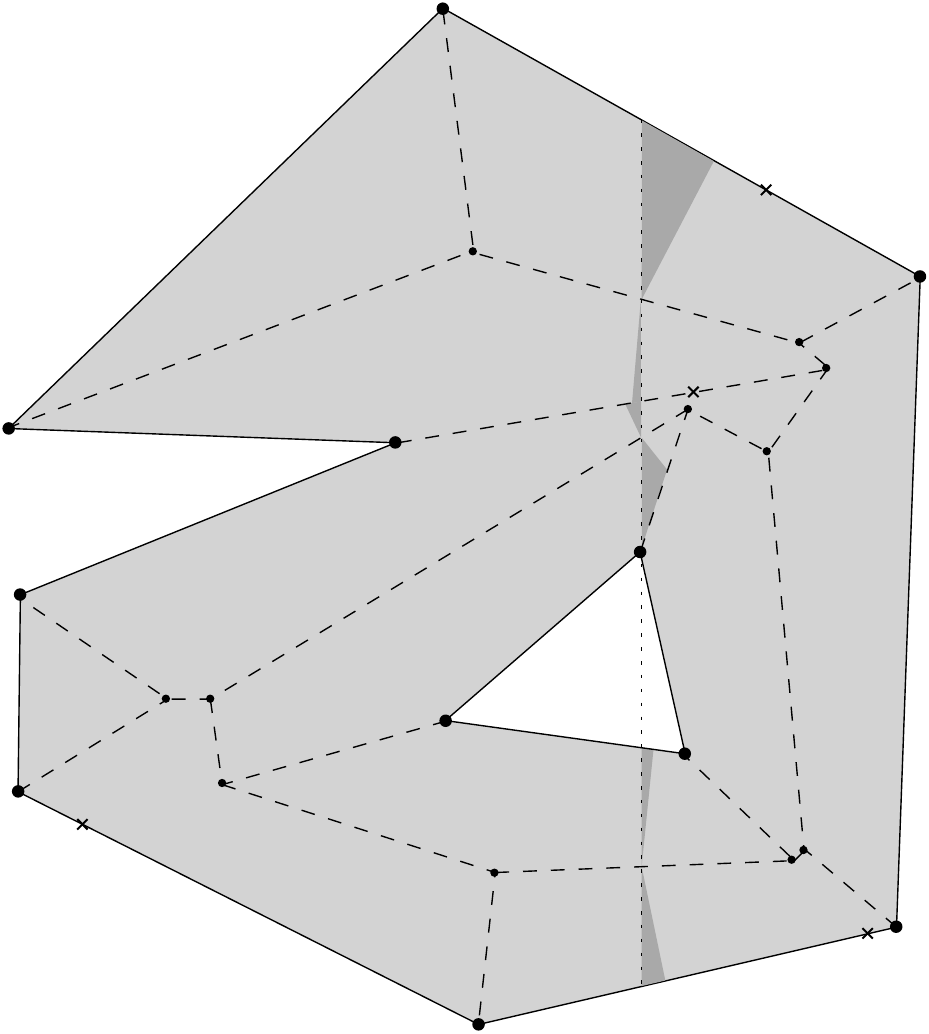}
		\caption{First vertical cut.}\label{eg:7}
	\end{subfigure}
	\centering
	\begin{subfigure}[b]{0.45\textwidth}
		\centering\includegraphics[width=.7\textwidth]{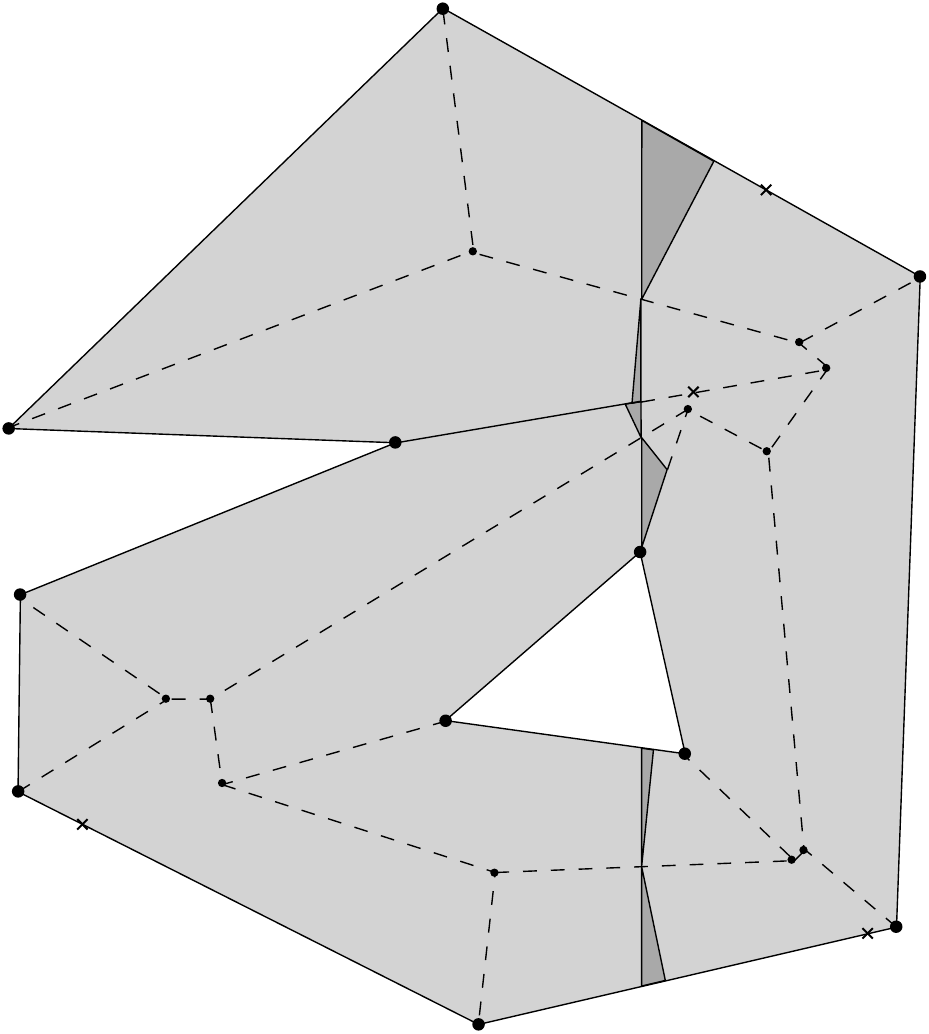}
		\caption{Subdivision induced by the first vertical cut.\label{fig:cut1}}
	\end{subfigure}
	\hspace{5ex}	
	\begin{subfigure}[b]{0.45\textwidth}
		\centering\includegraphics[width=.7\textwidth]{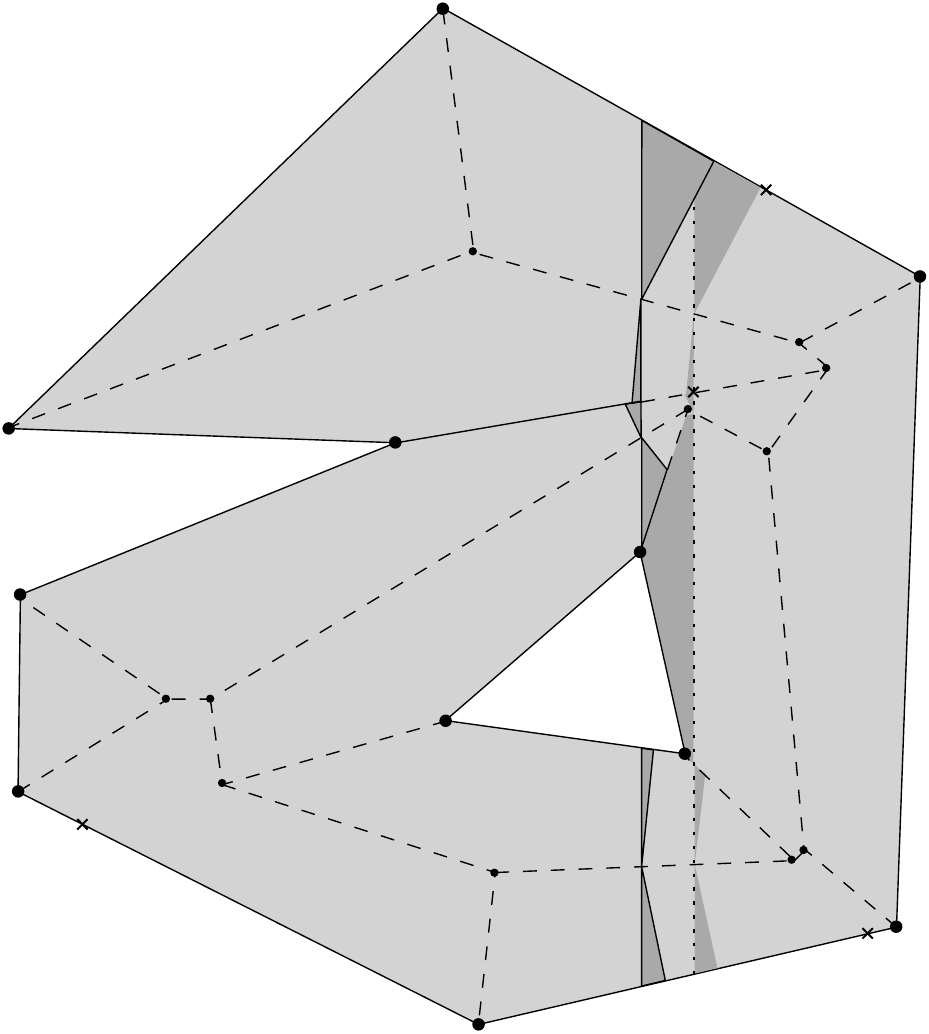}
		\caption{Second vertical cut.}\label{eg:7}
	\end{subfigure}
	\begin{subfigure}[b]{0.45\textwidth}
		\centering\includegraphics[width=.7\textwidth]{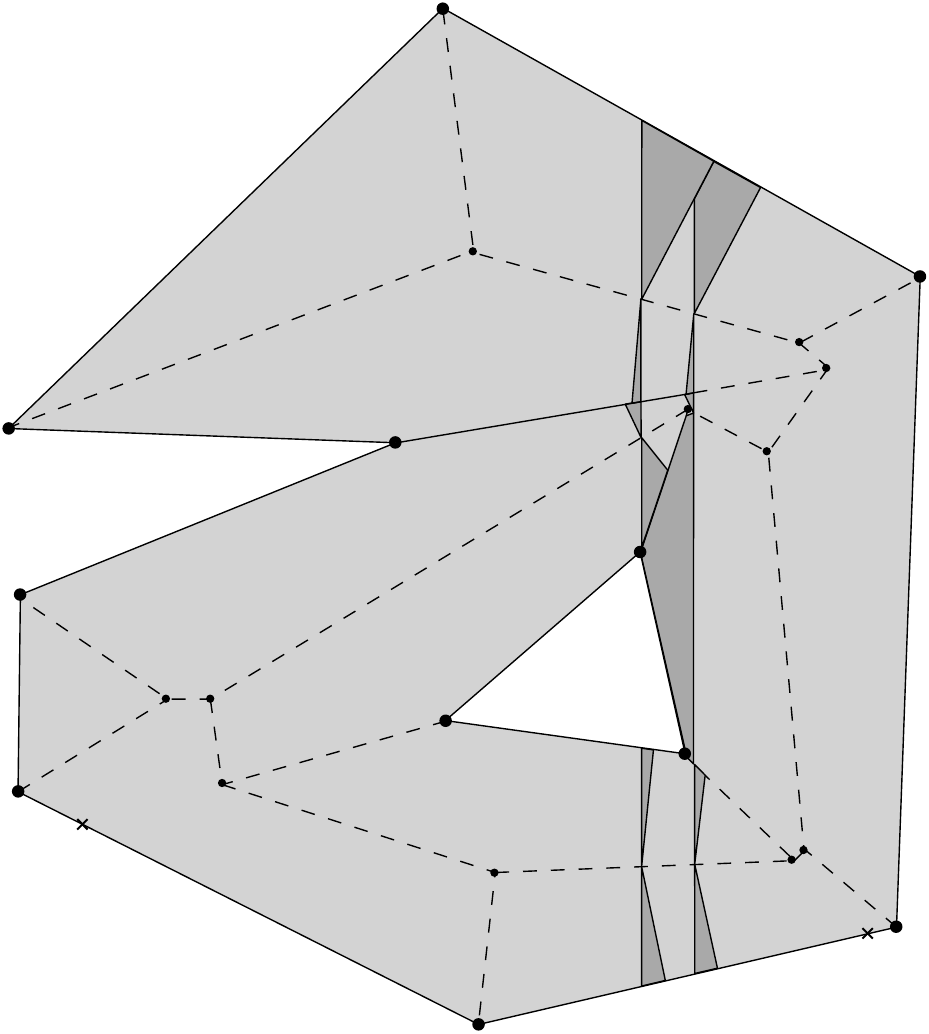}
		\caption{Subdivision induced by the second vertical cut.}\label{eg:7}
	\end{subfigure}
	\hspace{5ex}	
	\begin{subfigure}[b]{0.45\textwidth}
		\centering\includegraphics[width=.7\textwidth]{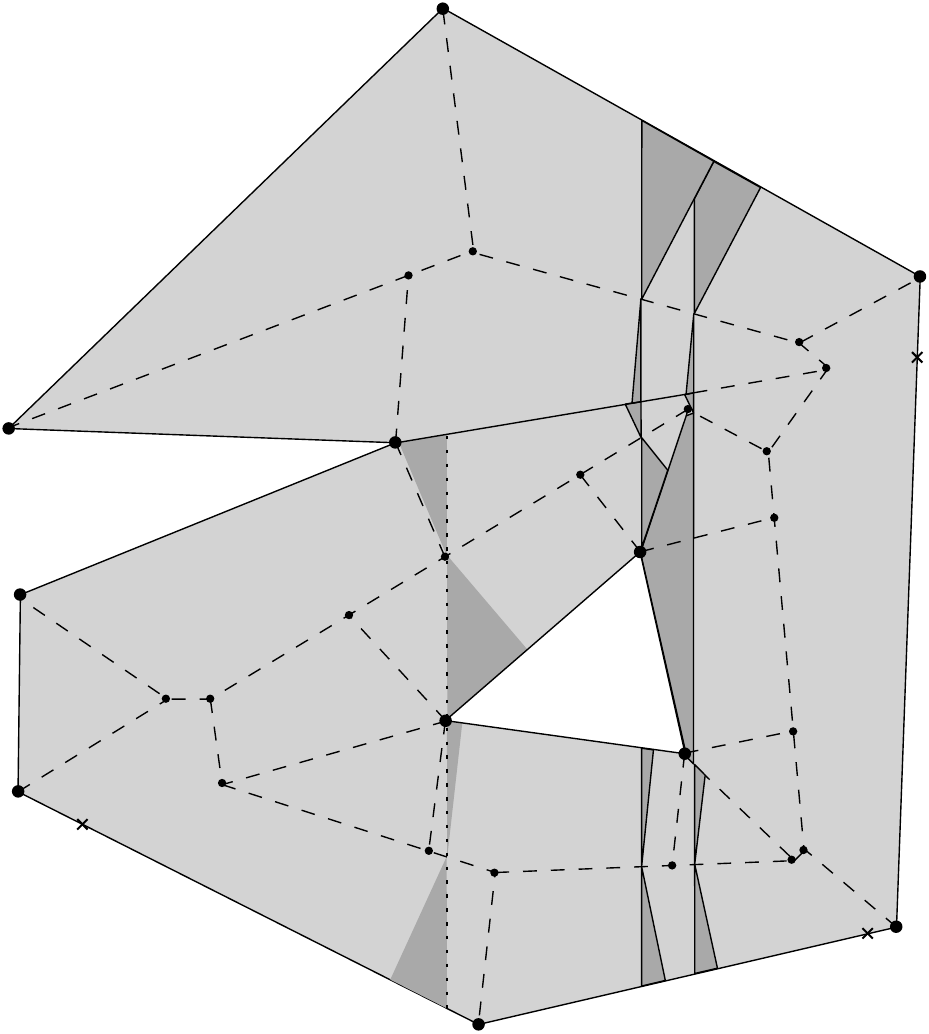}
		\caption{Third vertical cut.}\label{eg:7}
	\end{subfigure}
	\caption{The vertical subdivision. (Continued in \figurename~\ref{fig:example2}.) \label{fig:example}}
\end{figure}

\section{Cutting between valleys}
\label{sec:valley}

\subsection{Algorithm}

In this section, we describe the second stage of the algorithm. The corresponding procedure is called {\scshape Divide-Valley}, and its pseudocode is supplied in Algorithm~\ref{alg:valley}. Let $\cell_i$ be a cell of $\subdiv$ constructed by {\scshape Divide-Vertical} on which we call {\scshape Divide-Valley}. This cell $\cell_i$ is not empty and is not a wedge, as they are handled by brute force by {\scshape Divide-Vertical}, so by Lemma~\ref{lem:vertical}, it does not contain any reflex vertex in its interior. Let $R_i$ denote the set of valleys that conflict with $\cell_i$, and let $r_i$ denote its cardinality. The {\em extended valley} $e'$ corresponding to a valley $e \in R_i$ is the segment obtained by extending $e$ until it meets the boundary $\partial \cell_i$ of the cell. By Lemma~\ref{lem:valley}, the valley $e$ must meet $\partial \cell_i$, so we only need to extend it in one direction so as to obtain $e'$. As $\cell_i$ does not contain any motorcycle vertex in its interior, it implies that the extended valleys of $\cell_i$ do not cross. By Lemma~\ref{lem:valley}, the cell $\cell_i$ is simply connected, so the extended valleys form an outerplanar graph with outer face $\partial \cell_i$. (See \figurename~\ref{fig:outerplanar}.)
\begin{figure}
	\centering \includegraphics[width=\textwidth]{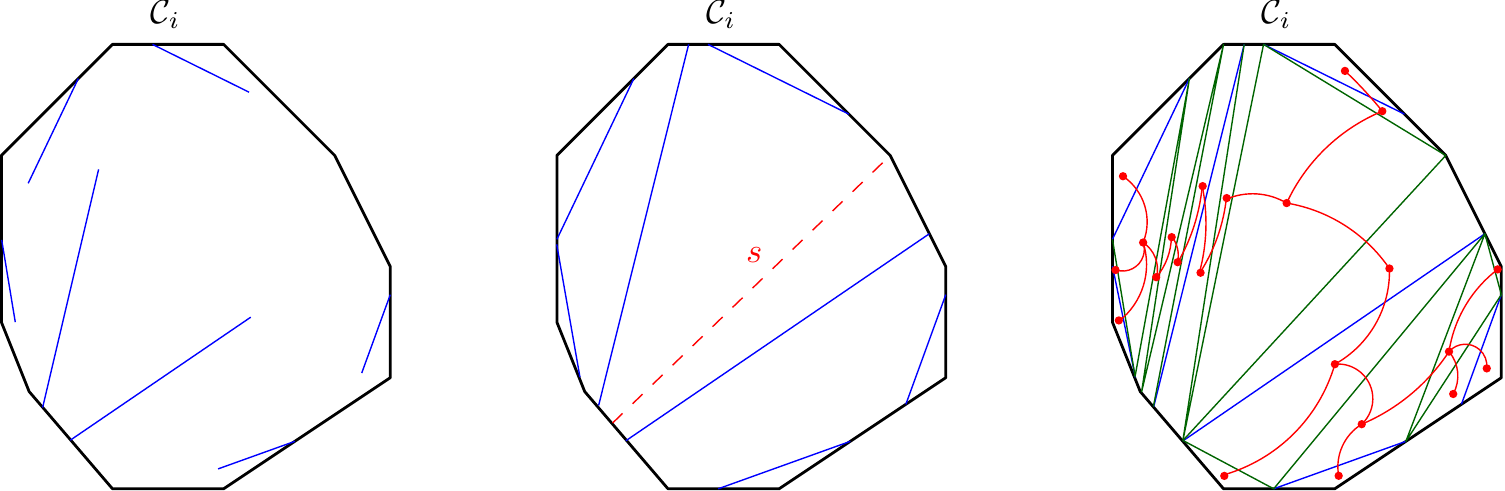}
	\caption{(Left) The cell $\cell_i$ and the conflicting valleys. 
		(Middle) The extended valleys, and a balanced cut.
		(Right) The triangulation and its dual graph.
		\label{fig:outerplanar}}
\end{figure}

At this stage of the algorithm, the cells are simply connected, so we record each cell $\cell_i$ using a single face list.  We do not need vertex conflict lists, as the cells do not conflict with any vertex. We do not need to store the valley conflict list $R_i$ either, as we can obtain it in linear time from the face list. 

If $\cell_i$ conflicts with at least one valley, we first construct a {\em balanced cut}, which is a chord $s$ of $\partial \cell_i$ such that there are at most $2r_i/3$ extended valleys on each side of $s$. (See \figurename~\ref{fig:outerplanar}, middle.) The existence and the algorithm for computing $s$ are explained below, in Lemma~\ref{lem:balanced}, but we first describe the rest of the algorithm. This balanced cut plays exactly the same role as the vertical edges $s_1,\dots,s_q$ along the cutting line that were used in {\scshape Divide-Vertical}. So we insert $s$ as a new boundary segment, we compute its lifted version $\hat s$, and at each crossing between $s$ and $\ske'$, intersects the descent paths as new boundary edges.

We repeat this process recursively, and we stop recursing whenever a cell does not conflict with any valley. All the structural results in Section~\ref{sec:vertical} still hold, except that now a cell is sandwiched between two balanced cuts, which can have arbitrary orientation, instead of the lines $\ell_i^-$ and $\ell_i^+$.

So now we assume that we reach a leaf $\cell_i$, which does not conflict with any valley. By Lemma~\ref{lem:convexity}, this cell $\cell_i$ must be convex. As valleys are the only reflex edges of $\terrain$, its restriction $\hat \cell_i$ above $\cell_i$ is convex. Hence, it is the lower envelope of the supporting planes of its faces. These faces are obtained in $O(c_i)$ time from the face lists, and the lower envelope can be computed in $O(c_i \log c_i)$ time algorithm using any optimal 3D convex hull algorithm.~\footnote{Although it would not improve the overall time bound of our algorithm, we can even compute $\hat \cell_i$ in $O(c_i)$ time using a linear-time algorithm for the medial axis of a convex polygon~\cite{aggarwal}: First construct the polygon on the $xy$-plane that is bounded by the traces of the supporting planes of the faces of $\hat \cell_i$, then compute its medial axis, and construct its intersection with $\cell_i$.}  We project $\hat \cell_i$ onto the $xy$-plane and we obtain the restriction $\ske'_i$ of $\ske'$ to $\cell_i$.

\begin{algorithm}
\caption{Cutting between valleys} \label{alg:valley}
\begin{algorithmic}[1]
\Procedure{Divide-Valley}{$\cell_i$}
   	\If{no valley conflicts with $\cell_i$}
  		\State Compute $\ske' \cap \cell_i$ as a lower envelope of planes.
		\State \Return
   	\EndIf
	\State Build the list of all valleys conflicting with $\cell_i$.
   	\State Construct a balanced cut $s$ as in Lemma~\ref{lem:balanced}.
	\State Construct the vertical slab $H$ through $s$.
   	\State Construct $\hat s$ as the lower envelope of the slabs intersecting $H$.
	\State Trace within $\cell_i$ the two or three steepest descent paths from 
			each vertex of $\hat s$.
	\State Update the partition $\subdiv$ using $s$ and the descent paths as new boundaries.
	\For{each child cell $\cell_j$ of $\cell_i$}
		\State Construct the data-structure for $\cell_j$.
		\State Call {\scshape Divide-Valley}$(\cell_j)$.
     	\EndFor
\EndProcedure
\end{algorithmic} 
\end{algorithm}

\subsection{Analysis}
It remains to analyses this algorithm, and prove the existence of a balanced cut. 
\begin{lemma}\label{lem:balanced}
Given a simply connected cell $\cell_i$ that does not conflict with any motorcycle vertex, and that conflicts with at least one valley, and given the face list of $\cell_i$, we can compute a balanced cut of $\cell_i$ in time $O(c_i \log c_i)$.
\end{lemma}
\begin{proof}
By Lemma~\ref{lem:structuresize}, the cell $\cell_i$ has $O(c_i)$ edges. We obtain the list $R_i$ of valleys conflicting with $\cell_i$ in $O(c_i)$ time by traversing the face list. Let $e_1,\dots,e_q$ denote these valleys. We first compute the set of extended valleys $R'_i=\{e'_1,\dots,e'_q\}$. The set $R'_i$ can be obtained in $O(c_i)$ time by traversing $\partial \cell_i$. We start at an arbitrary vertex of $\cell_i$, and each time we encounter the lower endpoint of a valley, we push the valley into a stack. At each edge $u$ of $\cell_i$ that we traverse, we check whether the extended valley $e'_j$ at the top of the stack meets it, and if so, we draw $e'_j$, we pop it out of the stack, and we check whether the new edge at the top of the stack meets $u$. 

Now we consider the outerplanar graph obtained by inserting the chords of $R'_i$ along $\partial \cell_i$. (See \figurename~\ref{fig:outerplanar}, middle.) We triangulate this graph, which can be done in $O(c_i)$ time using Chazelle's linear-time triangulation algorithm~\cite{Chazelle}, or in $O(c_i \log c_i)$ time using simpler algorithms~\cite{4M}. We construct the dual of this triangulation. We subdivide any edge of the dual corresponding to an extended valley, and we assign weight one to the new node. The other nodes have weight zero.  This graph is a tree, with degree at most 3, so we can compute a weighted centroid $\omega$ in time $O(c_i)$~\cite{KarivHakimi}. This centroid is a node of the tree such that each connected component of the forest obtained by removing the centroid has weight at most $r_i / 2$.

If $\omega$ corresponds to an extended valley $e'_j$, we pick $s=e'_j$ as the balanced cut. It splits $R_i$ into two subsets of size at most $r_i / 2$. Otherwise, $\omega$ corresponds to a face of the triangulation, such that the three subgraph rooted at $c$ have weight at most $r_i / 2$. We cut along the edge $s$ of this triangular face corresponding to the subtree with largest weight. 
\end{proof}

Lemma~\ref{lem:balanced} plays the same role as Lemma~\ref{lem:onesubdivision} in the analysis of {\scshape Divide-Vertical}. At each level of recursion, the size of the largest conflict list $R_i$ is multiplied by at most $2 / 3$, so the recursion depth is still $O(\log r)$. A leaf cell $\cell_i$ is handled in $O(c_i \log c_i)$ time by computing a lower envelope of planes, as explained above. It follows that we can complete the second step of the subdivision, and compute $\ske'$ within each cell, in overall $O(n (\log n) \log r)$ time. Then Theorem~\ref{th:main} follows.

Our analysis of this algorithm is tight, as shown by the example in Section~\ref{sec:tightness}.

\begin{figure}
	\centering
	\begin{subfigure}{\textwidth}
		\centering \includegraphics[width=.5\textwidth]{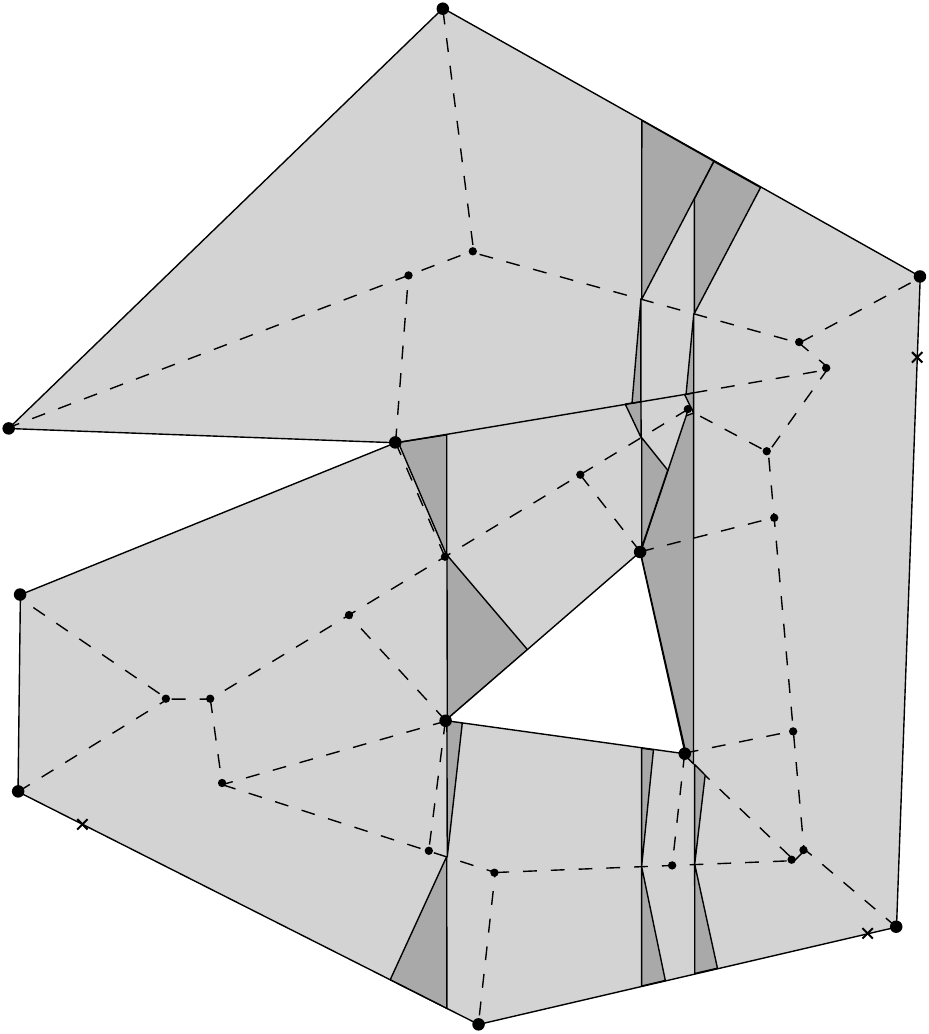}
		\caption{Vertical subdivision computed by {\scshape Divide-Vertical} on the same 
			polygon as in \figurename~\ref{fig:example}. The crosses along $\partial \poly$ are
			the terminal vertices of the motorcycle tracks. Note that the rightmost cell lies below two valleys, hence a non-trivial application of {\scshape Divide-Valley} is required.\label{fig:example2_1}}
	\end{subfigure}
	\begin{subfigure}{\textwidth}
		\centering \includegraphics[width=.5\textwidth]{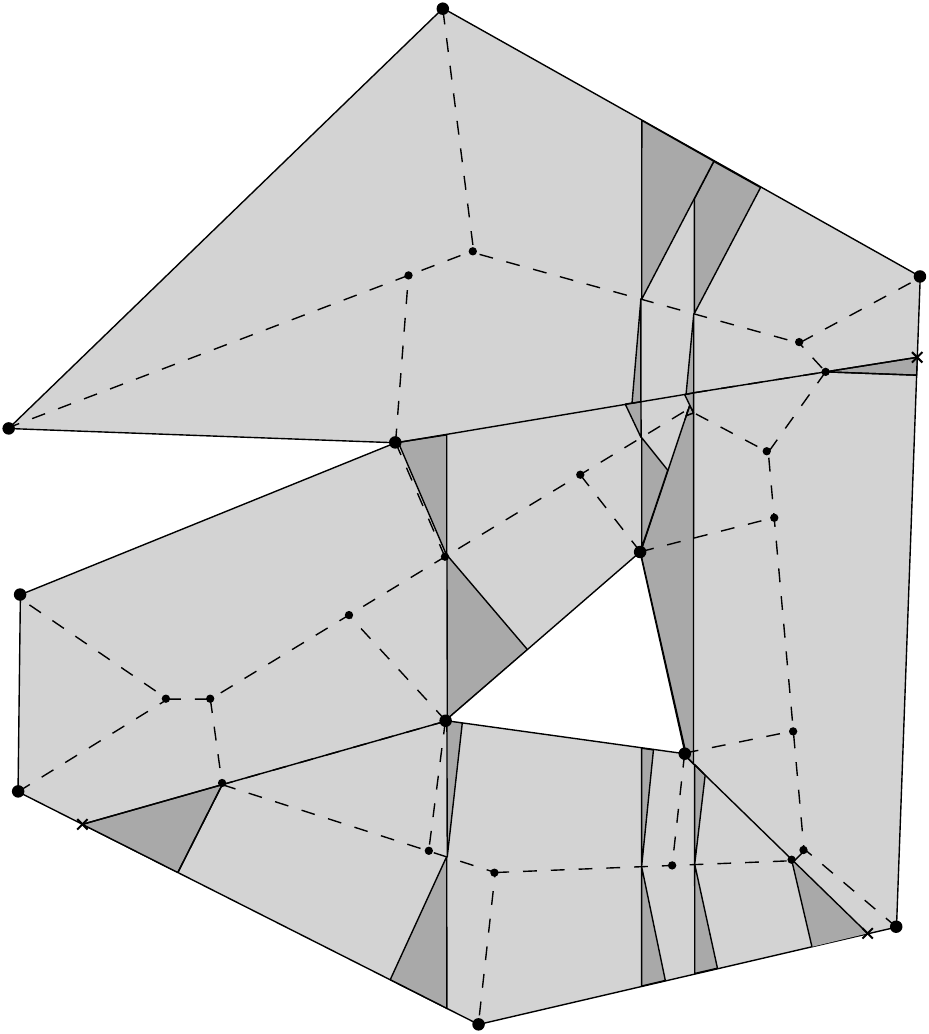}
		\caption{Final subdivision computed by {\scshape Divide-Valley}. 
		\label{fig:example2_2}}
	\end{subfigure}
	\caption{The result of the two stages of subdivision.\label{fig:example2}}
\end{figure}

\section{Degenerate cases}
\label{sec:degen}
As discussed in Section~\ref{sec:prelim}, the description and analysis of our algorithm was given for polygons in general position. Here we briefly explain why our result generalizes to arbitrary polygons.

As explained in the article by Eppstein and Erickson~\cite{eppstein}, almost all degeneracies can be treated by standard perturbation techniques, replacing high degree nodes with several nodes of degree 3. The only difficult case is when two or more valleys meet, and generate a new valley. In the induced motorcycle graph, this situation is represented by two or more motorcycles colliding, and generating a new motorcycle~\cite{huber2}.

So in degenerate cases, we assume that the exact induced motorcycle graph has been computed. It can be done in time $O(r^{17/11+\varepsilon})$ for any $\varepsilon>0$, using Eppstein and Erickson's algorithm~\cite{eppstein}. Then the problem becomes one of computing a lower envelope of slabs. Standard perturbation techniques apply to this problem~\cite{SOS}, so our non-degeneracy assumptions are valid. 

The only difference with the non-degenerate case is that now, instead of having each valley adjacent to a reflex vertex, the valleys form a forest, with leaves at the reflex vertex. So a descent path may be a polyline with arbitrarily many vertices. Thus, when we perform a vertical cut, we cannot necessarily trace a descent path in constant time. However, we can trace it in time proportional to its size, and its edges become cell boundaries. The subdivision can be updated in amortized $O(\log n)$ time for each such edge, as we update the partition by plane sweep. So the extra contribution to the overall running time is $O(n \log n)$.

\section{Tightness of analysis}
\label{sec:tightness}

We give an example to demonstrate that for this algorithm the analysis is tight. Consider a polygon $\mathcal{P}$ where, on the left hand side, we have a convex chain of $\Omega (n)$ near-vertical edges. Along the top boundary of $\mathcal{P}$ we have $\Omega (r)$ small reflex dips pointing downwards. See \figurename~\ref{fig:tight} for an example with a convex chain of size 4, and 5 reflex dips. The straight skeleton faces corresponding to each edge of the convex chain to the left of the polygon extend deep into the polygon. Each time we make a vertical cut to the right of all other vertical cuts previously made, it will cross through all faces of the chain, hence all the slabs must be provided to the lower envelope calculation. It then follows that Algorithm~\ref{alg:vertical} spends $\Omega (n (\log n) \log r)$ time as it computes $\Omega(\log r)$ lower envelopes of size $\Omega(n)$. 

\begin{figure}
\centering
\includegraphics[scale=.75]{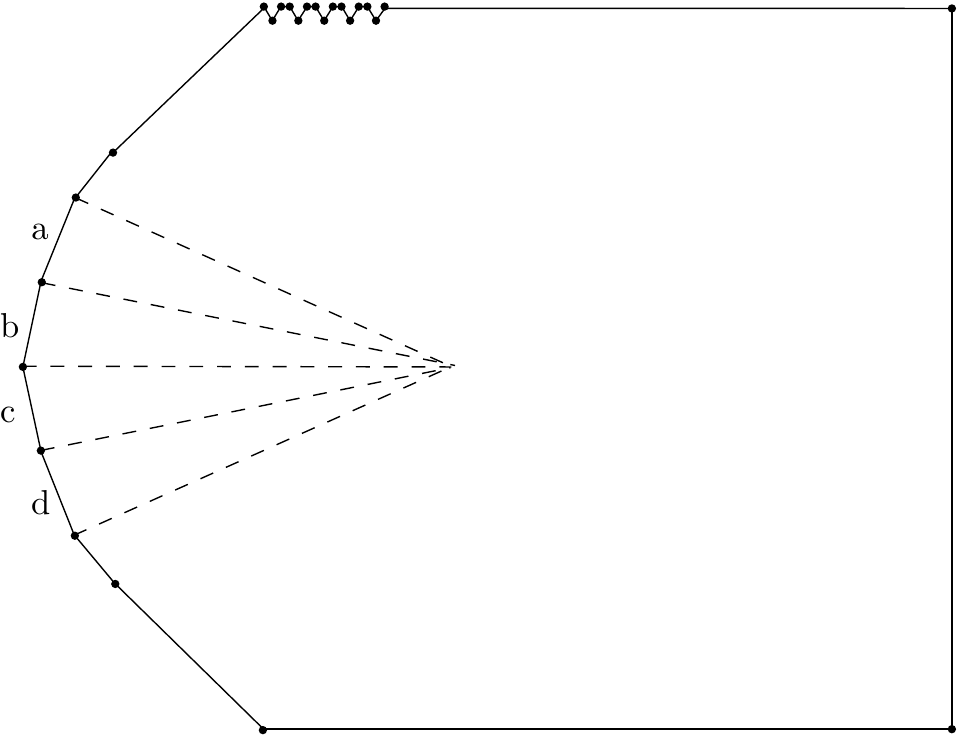}
\caption{Tight example. For vertical cuts that are introduced from left to right, the four slabs corresponding to $e_1, e_2, e_3, e_4$ conflict with the cuts. \label{fig:tight} }
\end{figure}

\bibliographystyle{plain}
\bibliography{references}

\end{document}